\documentclass[english,12pt]{article}
\usepackage[latin9]{inputenc}
\usepackage{calc}
\usepackage{mathrsfs}
\usepackage{mathtools}
\usepackage{amsmath}
\usepackage{amsthm}
\usepackage{amssymb}
\usepackage[round,compress,sort]{natbib}
\usepackage{svg}

\usepackage{url} 
\usepackage{soul}

\usepackage{enumitem}
\usepackage{hyperref}
\hypersetup{hidelinks}
\usepackage[ruled,vlined]{algorithm2e}
\makeatletter
\usepackage{setspace}
\usepackage{booktabs} % For better horizontal rules
\usepackage{caption} % For customizing captions
\usepackage{siunitx} % For improved decimal alignment
\usepackage{threeparttable}
\usepackage{bbm}
\usepackage{xcolor}
\usepackage{sectsty}
\usepackage{threeparttable}
\usepackage{longtable}

\newcommand{\diag}{\mathrm{diag}}

\theoremstyle{plain}
\newtheorem{thm}{\protect\theoremname}
\theoremstyle{plain}
\newtheorem{assumption}[thm]{\protect\assumptionname}
\theoremstyle{definition}
\newtheorem{defn}[thm]{\protect\definitionname}
\theoremstyle{remark}
 
\theoremstyle{plain}
\newtheorem{lem}[thm]{\protect\lemmaname}
\theoremstyle{plain}

\theoremstyle{plain}
\newtheorem{cor}[thm]{\protect\corollaryname}
\newtheorem{proposition}{Proposition}
\theoremstyle{plain}
\newtheorem{example}{Example}
\@ifundefined{date}{}{\date{}}
\usepackage{amsthm}\usepackage{mathrsfs}
\usepackage{bm}

\usepackage{authblk}
\usepackage{enumitem}
\usepackage{babel}
\usepackage{float}
\usepackage{graphicx}
\usepackage{tikz}
\usepackage{xcolor}% for \textcolor
\usepackage{soul}% for \hl (highlight)
\definecolor{myblue}{HTML}{0B5394} % a darker blue
\usetikzlibrary{arrows.meta, positioning}

\makeatother

\usepackage{babel}
\providecommand{\assumptionname}{Assumption}
\providecommand{\corollaryname}{Corollary}
\providecommand{\definitionname}{Definition}
\providecommand{\lemmaname}{Lemma}
\providecommand{\questionname}{Question}
\providecommand{\theoremname}{Theorem}

\usepackage{bm}
\newcommand{\vect}[1]{\bm{#1}} % vectors
\newcommand{\E}{\mathbb{E}}

\newcommand{\Var}{\mathrm{Var}}
\newcommand{\Cov}{\mathrm{Cov}}
\newcommand{\Prob}{\mathbb{P}}
\newcommand{\indep}{\perp\!\!\!\perp}

\begin{document}
\title{\bf  Shrinkage priors for Bayesian Substitute Confounders}
\author[1]{Yordan P. Raykov}
\author[2]{Hengrui Luo}
\author[3]{Justin D. Strait}
\author[1]{Wasiur R. KhudaBukhsh}
\affil[1]{School of Mathematical Sciences, University of Nottingham, Nottingham, UK} 
\affil[2]{Department of Statistics, Rice University, USA; Lawrence Berkeley National Laboratory, USA}
\affil[3]{Statistical Sciences Group, Los Alamos National Laboratory, USA}

\maketitle
\begin{abstract}

Multi-cause observational studies contain information about unmeasured confounding through the dependence structure among causes. However, literal imputation of the unobserved confounder is often more complex than learning a lower-dimensional substitute score that preserves the shared assignment variation needed for stable causal adjustment. The \emph{deconfounder} \citep{wang2019blessings} and related substitute confounder methods exploit this idea, but flexible
assignment models can fit the joint distribution of the causes while producing scores that over-encode the treatment vector, collapse overlap, or capture single cause variation. We develop a Bayesian factor assignment framework for learning sparse substitute confounders that retain coarse multi-cause dependence with shrinkage priors. The theory is stated at the level of posterior concentration, factor score contraction, and overlap preserving assignment geometry, and therefore does not rely on a particular shrinkage prior. Under these conditions, the proposed regression-adjusted estimators are consistent for mean potential outcomes when the corresponding latent variable identification assumptions hold. Shrinkage priors provide a natural tool for latent structural learning: they favour low-dimensional factors supported by multiple causes, discourage effectively single-cause factors, and induce an ordering of the latent factors through progressive shrinkage. Synthetic experiments illustrate the roles of signal strength, outcome validity, and geometry-aware regularization. In an Alzheimer's Disease Neuroimaging Initiative (ADNI) baseline analysis, sparse substitute scores recover much of the adjustment obtained by directly conditioning on invasive cerebrospinal-fluid biomarkers, while collapse diagnostics identify when fitted factors reduce to individual observed measurements.
\end{abstract}
\tableofcontents
%\newpage
\def\spacingset#1{\renewcommand{\baselinestretch}%
{#1}\small\normalsize} \spacingset{1.} 
\section{Introduction}

Modern observational studies often record many causes whose effects are simultaneously of interest, as in genomics \citep{yu2006unified}, neuroimaging \citep{mueller2005alzheimer}, and digital health \citep{evers2020real}.  The same multiplicity that makes these studies challenging can also be informative: dependence among non-randomly assigned causes may reveal shared variation in the assignment mechanism, and hence carry partial information about unmeasured common causes.  The inferential task is not necessarily to impute the latent confounder itself: it is often more plausible to learn a lower-dimensional substitute score that preserves the shared assignment variation needed for stable adjustment.  However, high-dimensional assignment laws are rarely well supported in all directions, so weighting, regression adjustment, and extrapolation can become unstable for treatment contrasts that lie in weak overlap regions. Regularization, or causal targets that deliberately avoid such directions, is therefore part of the identification and estimation problem rather than merely a numerical convenience.

Our starting point is the deconfounder framework of \citet{wang2019blessings,wang2019rejoinder}.  Its central insight is that, when many causes are observed, their joint assignment distribution may contain information about unmeasured common causes.  This is formalized by fitting an \emph{assignment model} for the causes, in which a low-dimensional latent variable explains their dependence, and then using the fitted latent variable as a substitute adjustment in the outcome model.  Related latent variable approaches develop the same principle in implicit or variational forms \citep{tran2017implicit,ranganath2018multiple}. However, with many causes and flexible assignment models, the learned latent representation can behave pathologically.  Factors may load on only a single cause, encode nearly injective functions of the treatment vector, or remove so much residual variation that the conditional variability of the causes given the latent variable collapses.  Conditioning on such a representation can destroy overlap and invalidate causal adjustment, even when a useful low-dimensional confounding signal exists \citep{ogburn2019comment,damour2019multi}.  At the same time, a growing literature shows that shared latent confounding across many causes is generally insufficient for point identification without additional restrictions on the outcome model, the causal target, or auxiliary variables such as instruments or negative controls
\citep{damour2019multi,kong2022identifiability,veitch2019using}.  Multi-cause settings provide extra leverage, but do not remove the need for structure.

We take a complementary view and focus on stability of estimation under substitution, not identification with an oracle latent confounder alone.  A coarse substitute score can denoise shared assignment variation, whereas an over-rich score can over-condition on the treatment vector and collapse overlap. Thus, the central theoretical question is not simply whether an oracle latent confounder would identify a causal functional, but whether replacing it by an estimated substitute score produces controlled error for the causal target being reported.  For full mean potential outcome targets, we show that regression-adjusted functionals are stable when the outcome regression is calibrated at the fitted scores and the substitute scores contract toward an admissible latent confounder.  When some treatment directions remain poorly supported after adjustment, we instead define projected or regularized targets that avoid those directions. This leads to a geometric view of substitute confounding: the assignment model not only produces adjustment scores; it also identifies directions of the treatment space that are persistently explained by shared latent factors.  These are directions of high confounding uncertainty and should be projected out or shrunk in downstream causal estimation.  This connects substitute confounder learning to regularized causal estimators, including ridge-regularized balancing and regression \citep{bruns2025augmented}.

In this paper, we propose a geometric substitute confounder approach with two components. On the \emph{identification} side, we separate two targets. Full mean potential outcome targets require latent ignorability and latent-level positivity. When effects along the latent confounding directions are not identified, we instead define \emph{projected} causal estimands that depend only on treatment variation orthogonal to the confounding subspace. These restrictions accommodate nonparametric confounding pathways while making explicit which treatment directions are identified and which are excluded or regularized. On the \emph{estimation} side, we show how to learn a stable, overlap-preserving substitute confounder from many causes using sparse Gaussian factor assignment models with shrinkage priors on the factor loadings.

We implement this idea using global-local and ordered shrinkage priors on the factor-loading matrix \citep{bhattacharya2011sparse,legramanti2020bayesian}.
These priors concentrate posterior mass on low-dimensional shared multi-cause structure, suppress spurious near single-cause factors, and help preserve conditional variability of the causes given the factors.  Since factor models are identifiable only up to rotations, and continuous shrinkage priors do not
produce exact zeros, we formulate the needed assignment stage property through effective oriented loadings and an explicit \emph{factor score contraction} condition. Substituting these scores into the identified outcome models yields consistent
frequentist plug-in and Bayesian posterior predictive estimators, with bias controlled by the contraction rate.  Thus, outcome restrictions define the causal target, while shrinkage regularizes the assignment model toward the effective sparsity and contraction properties needed for stable,
geometry-aware estimation.

Section~\ref{sec:prelims-deconf} introduces the multi-cause potential outcome setup, the substitute confounder premise via assignment models, and the stability lens that motivates geometry-aware estimands. Section~\ref{sec:sparse-assignment} formalizes failure modes of flexible assignment models, develops sparse factor assignment structure that preserves conditional variability, and states plug-in stability results linking factor score error to causal error. Section~\ref{sec:bayesian-assignment} places the assignment stage in a Bayesian framework, describes shrinkage priors and the induced posterior confounding geometry, and proves consistency of posterior averaged causal functionals under factor score contraction. Section~\ref{sec:experiments} presents synthetic experiments, an adaptive shrinkage study connecting the learned geometry to regularized causal
estimation, and an ADNI case study evaluating adjustment recovery and factor collapse diagnostics.

\section{Related work}
\label{sec:related}

In the causal inference literature, a long-standing response to unmeasured confounding is to quantify robustness rather than seek point identification. An established approach is sensitivity analysis which formalizes how large an unobserved bias would need to be to overturn conclusions, beginning with matched-study formulations \citep{rosenbaum1983assessing} and early bounding arguments in epidemiology \citep{cornfield1959smoking}. Modern variants provide interpretable summary measures such as E-values \citep{vanderweele2017sensitivity} and connect robustness to omitted-variable bias calculations that benchmark unmeasured confounding against observed covariates \citep{cinelli2020making}.
In parallel, partial-identification approaches derive sharp bounds under minimal assumptions \citep{manski1990nonparametric}. These tools are essential for assessing sensitivity and robustness, but they do not by themselves explain how to construct adjustment variables or estimands that remain stable when the treatment is high-dimensional and overlap is weak.

A different route would be to introduce auxiliary variables (e.g. proxies, negative controls, or instruments) to recover information about latent confounders. Negative-control and proxy-variable methods identify causal effects by introducing observed proxies for the unmeasured confounder and solving a confounding-bridge
equation; identification requires the associated conditional expectation operator to be injective, often formulated as a completeness condition \citep{miao2018identifying}. The broader negative-control literature discusses how negative controls can be deployed in observational settings \citep{tchetgen2020introduction,jin2025directed}. Related work gives graphical criteria for proxy sufficiency \citep{kuroki2014measurement,veitch2019using}, while
instrumental variable methods rest on exclusion and relevance assumptions \citep{angrist1996identification}. These approaches can deliver point identification, but their validity hinges on structural assumptions about the
auxiliary variables that are typically untestable and often difficult to justify in applications with complex treatment mechanisms.

Multi-cause regimes motivate a third perspective: dependence among many causes may reflect shared latent variation that simultaneously drives multiple components of the treatment vector. The deconfounder framework
\citep{wang2019blessings,wang2019rejoinder} operationalizes this idea by fitting an \emph{assignment model} in which a low-dimensional latent variable renders the causes conditionally independent, and then using factor score summaries as a substitute confounder for outcome adjustment. Closely related latent variable adjustment approaches include implicit causal models and neural architectures for high-dimensional observational and genetic settings \citep{tran2017implicit,ranganath2018multiple}. This line of work has also warranted rigorous examination of the underlying assumptions required for the deconfounder framework. \citet{ogburn2019comment} emphasize that core assumptions (notably, the absence of unmeasured single-cause confounders) are fundamentally untestable from observed data, and \citet{damour2019multi} show that multi-cause structure alone does not guarantee point identification without additional restrictions on the outcome or auxiliary variables. Identification results that do hold typically require strong asymptotic structure, such as latent confounders affecting infinitely many causes \citep{kong2022identifiability}. Extensions adapt the multi-cause latent variable idea to time series \citep{bica2020estimating}, spatial regimes \citep{khot2025spatial}, as well as others, but the central tension remains: a flexible factor model can fit the joint distribution of $\vect X$ while producing latent summaries that are unhelpful for causal adjustment, either by over-encoding $\vect X$ (destroying overlap) or by capturing the wrong low-dimensional structure.

These observations motivate a stability view of substitute confounding: beyond whether an ideal latent adjustment would identify a target, one must control what happens when that latent variable is replaced by a fitted score and when treatment directions have weak support. This places our work alongside regularized causal inference in high dimensions, where conditioning and overlap can make estimators poorly conditioned even when confounders are observed. A large body of literature develops valid inference with high-dimensional nuisance estimation, including post-selection methods \citep{belloni2014inference} and orthogonalized estimators in double/debiased machine learning \citep{chernozhukov2018double}, with related analysis for propensity-based adjustment \citep{farrell2015robust} and for heterogeneous effects via causal forests \citep{wager2018estimation}. More directly connected to our geometric framing, ridge-regularized augmented balancing weights interpret stability as a metric/penalty choice and demonstrate improved finite-sample tradeoffs when sparsity is implausible \citep{bruns2025augmented}. We build on this viewpoint, but shift the source of geometry: in our setting the metric is induced by the learned confounding structure in the assignment model, and regularization is motivated by overlap preservation and by controlling the substitution error incurred when $Z$ is replaced by $\hat Z$. Our assignment models draw on sparse Bayesian factor analysis, especially global-local and ordered shrinkage priors that adaptively regularize
factor loadings \citep{bhattacharya2011sparse,legramanti2020bayesian}. We use
this literature as a modeling device for overlap-preserving substitute scores; the causal role of shrinkage is developed in Sections~\ref{sec:sparse-assignment}
and~\ref{sec:bayesian-assignment}.

\section{Preliminaries}
\label{sec:prelims-deconf} 
 
Suppose we have $n$ independent and identically distributed (i.i.d.) units from a superpopulation distribution
$\Prob_{\vect X,\vect V, Y}$. For unit $i\in\{1,\dots,n\}$ we observe
\begin{align}
  (\vect x_i,\vect v_i,y_i)\in \mathcal X\times\mathcal V\times\mathbb R,
\end{align}
where $\vect x_i=(X_{i1},\dots,X_{im})^\top\in\mathcal X\subseteq\mathbb R^m$ is a vector of $m$ causes, $\vect v_i\in\mathcal V\subseteq\mathbb R^d$ are observed pre-treatment covariates, and $ y_i\in\mathbb R$ is the outcome. Bold uppercase letters $\vect{X}$ denote random vectors and bold lowercase letters denote their realizations $\vect{x}$; letters not in bold indicate scalars. The notation $\vect{Y}\indep \vect{X}\mid \vect{V}$ denotes conditional independence between random vectors $\vect{Y}$ and $\vect{X}$ given the random vector $\vect{V}$, and $\Prob(\cdot)$, $\E[\cdot]$, $\Var(\cdot)$, and $\Cov(\cdot)$ denote probability, expectation, variance, and covariance, respectively.
 
We adopt the potential outcome framework of \citet{rubin1974estimating},
but now with a \emph{vector} of causes. For each treatment regime $\vect x\in\mathcal{X}$ we posit a corresponding potential outcome
\begin{align}
Y^{(\vect{x})} : (\Omega,\mathcal{A}) \longrightarrow (\mathbb{R},\mathcal{B}(\mathbb{R})),  \qquad \vect{x}\in\mathcal{X},
\end{align}
interpreted as the response we would observe \emph{had all} causes been set to $\vect{X}=\vect{x}$. Thus $\{Y^{(\vect x)}:\vect x\in\mathcal X\}$ is a family of random variables indexed by treatment regimes $\vect x$, all defined on the same probability space as $\vect X$. The observed outcome is linked to the potential outcomes through the usual consistency relation \(Y = Y^{(\vect X)}.\)

We focus on the mean potential outcome and the average treatment effect (ATE):
\begin{align}
  \mu(\vect x)=\E\!\bigl[ Y^{(\vect x)}\bigr],
  \qquad
  \Delta(\vect x,\vect x')=\mu(\vect x')-\mu(\vect x),
  \qquad \vect x,\vect x'\in\mathcal X.
\end{align}
More generally, let
\[
\eta_{\vect x}(A):=\Prob\!\bigl(Y^{(\vect x)}\in A\bigr),
\qquad A\in\mathcal B(\mathbb R),
\]
denote the distribution of the potential outcome under treatment regime $\vect x$.
Other causal functionals follow analogously. For example, we may consider marginal effects of a single component $x_k$ holding the remaining causes fixed, or functionals of the entire distribution of $Y^{(\vect x)}$ such as quantiles or risk measures. These objects are standard in multivalued and multi-treatment settings and can be handled by replacing $\mu(\vect x)$ with the corresponding functionals \citep{imbens2000role,lopez2017estimation}.

\subsection{Classical identification via observed covariates}
\label{subsec:prelims-backdoor}

Classically, a standard route to identifying $\mu(\vect x)$ (and hence $\Delta(\vect x,\vect x')$) is adjustment on observed covariates $\vect V$ that block back-door paths \citep{pearl1995causal}. This relies on ignorability of the potential outcomes given $\vect V$ and joint overlap of the treatment vector within covariate strata. 

\begin{assumption}[Ignorability / unconfoundedness]
\label{assum:ignorability}
For every treatment level $\vect{x} \in \mathcal{X}$,
\begin{align}
Y^{(\vect{x})} \;\indep\; \vect{X} \;\bigm|\; \vect{V}.
\end{align}
\end{assumption}

\begin{assumption}[Positivity / overlap]
\label{assum:positivity}
For every treatment level $\vect{x}$ in the support of interest,
\begin{align}
p_{\vect{X}|\vect{V}}(\vect{x} \mid \vect{v}) > 0
\end{align}
for $\Prob_{\vect{V}}$-almost every $\vect{v}\in \mathcal{V}$. When needed for estimation, we strengthen this to the conditional density being bounded away from zero on the support of interest.
\end{assumption}

Under Assumptions~\ref{assum:ignorability}-\ref{assum:positivity}, together with consistency, the
back-door adjustment identifies $\eta_{\vect x}$ for $A\in\mathcal B(\mathbb R)$ via
\begin{equation}
\label{eq:eta-backdoor}
  \eta_{\vect x}(A)
  =
  \int_{\mathcal V}
    \Prob\!\bigl( Y\in A \mid \vect X=\vect x,\vect V=\vect v\bigr)\,
  d\Prob_{\vect V}(\vect v).
\end{equation}
This classical baseline is appropriate when (i) the observed covariates
$\vect V$ are sufficiently rich to block all back-door paths and (ii) the joint overlap condition is plausible. In high-dimensional multi-cause regimes, both conditions can be difficult to justify; this motivates substitute confounder approaches that leverage structure in the joint distribution of $\vect X$.

\subsection{Substitute confounders: representation and existence}
\label{subsec:deconfounder}

For unit-level statements below, we write $Y_i^{(\vect x)}$ for the potential outcome of unit $i$ under treatment regime $\vect x$. The deconfounder viewpoint \citep{wang2019blessings} is that when the treatment vector $\vect X_i$ is high-dimensional, dependence among causes may reflect latent variables that
simultaneously affect multiple causes. A substitute confounder approach posits a latent variable $\vect Z_i\in\mathcal Z$ such that, informally, $\vect X_i=(X_{i1},\dots,X_{im})^\top$ is generated from shared latent variation plus idiosyncratic noise. We write
\[
\vect X_{1:n}:=(\vect X_1,\dots,\vect X_n)
\]
for the full sample of treatment vectors. Classical adjustment
\eqref{eq:eta-backdoor} relies on observed pre-treatment covariates $\vect V_i$ to render treatment assignment ignorable; the substitute confounder strategy instead learns an adjustment variable from the dependence structure of the
causes themselves.

A common formalization is an assignment model for the causes with a latent variable $\vect Z_i$ such that the causes are conditionally independent given $\vect Z_i$:
\begin{align}
\label{eq:assignment-model}
  \vect Z_i &\sim q_\alpha,
  &
  p_\theta(\vect x_i\mid \vect Z_i=\vect z_i)
  &= \prod_{j=1}^m p_{\theta_j}(x_{ij}\mid \vect z_i),
\end{align}
where $q_\alpha$ denotes the latent law and $p_{\theta_j}$ denotes the cause-specific conditional law.  We use different symbols for these two parts of the assignment model because only the conditional product structure is used to define the substitute confounder representation.

The procedure consists of two stages.  In the first stage, the assignment model is fitted to the full treatment sample $\vect X_{1:n}$, yielding estimates or posterior draws of the common assignment parameters.  A score for unit $i$ is then computed from its own treatment vector using those sample-level fitted parameters, for example
\[
  \hat{\vect Z}_{i,nm}=f_{\hat\alpha,\hat\theta}(\vect x_i;\vect X_{1:n})
  :=\mathbb E_{\hat\alpha,\hat\theta}(\vect Z_i\mid \vect X_i=\vect x_i),
\]
or by the corresponding posterior mode.  Thus $\hat{\vect Z}_{i,nm}$ is not estimated from a single observation alone: the common map is learned from $\vect X_{1:n}$ and then evaluated at $\vect x_i$.  The fitted scores are therefore dependent across units through the common first-stage fit.  The large sample results below allow for this dependence by formulating score accuracy as an average condition over $i=1,\ldots,n$; sample splitting or cross-fitting could also be used if one wanted independent first- and second-stage scores. Throughout the frequentist asymptotic sections we use $\hat{\vect Z}_{i,nm}$ to emphasize that score accuracy may depend on both the number of units $n$ and the number of causes $m$; in Bayesian statements, $\hat{\vect Z}_i(\theta_X)$ denotes the draw-specific score under assignment parameter $\theta_X$. In Section~\ref{sec:sparse-assignment}, we propose a structured sparse factor assignment model to estimate $\hat{\vect Z}_{i,nm}$ in a way that preserves multi-cause structure while avoiding unnecessary overlap deterioration.

A key structural ingredient in \citet{wang2019blessings} is a Kallenberg construction of the assignment mechanism.  In our setting, such a representation is combined with an explicit exogeneity condition on the assignment noise (Definition~\ref{def:kallenberg} and Lemma~\ref{lem:kallenberg-implies-unconf}) to yield the weak unconfoundedness
statement needed for causal adjustment.

\begin{defn}[Kallenberg construction]
\label{def:kallenberg}
We say that $\vect X_i$ admits a \emph{Kallenberg construction} from the latent
variable $\vect Z_i$ if there exist mutually independent random variables
$U_{i1},\dots,U_{im}\sim \mathrm U(0,1)$, independent of $\vect Z_i$, and
measurable functions $f_1,\dots,f_m$ such that
\begin{align}
\label{eq:kallenberg_con}
  X_{ij}=f_j(\vect Z_i,U_{ij}),\qquad j=1,\dots,m.
\end{align}
\end{defn}

A representation of the form $\vect X_i=f(\vect Z_i,U_i)$ with $U_i$ mutually independent uniforms can be constructed to match a given \emph{assignment distribution}.  However, the additional independence condition
$U_i\indep \{Y_i^{(\vect x)}:\vect x\in\mathcal X\}\mid \vect Z_i$ is a \emph{causal} assumption on the full data law and cannot be deduced from the observed distribution of $\vect X$ alone; see \citet{ogburn2019comment} for a detailed discussion of this point.

The next lemma records the pointwise weak-unconfoundedness statement used in \citet{wang2019blessings} and first introduced in \citet{imbens2000role}.

\begin{lem}[Kallenberg construction implies unconfoundedness]
\label{lem:kallenberg-implies-unconf}
Fix a unit $i$. Suppose there exist mutually independent random variables $U_{i1},\dots,U_{im}\sim \mathrm{Unif}(0,1)$ and measurable functions $f_1,\dots,f_m$ such that
\[
  X_{ij}=f_j(\vect Z_i,U_{ij}) \quad \text{a.s. for all } j\in\{1,\dots,m\}.
\]
Assume moreover that $U_i:=(U_{i1},\dots,U_{im})$ is independent of
$\vect Z_i$ and is conditionally independent of the potential outcome table
\(\{Y_i^{(\vect x)}:\vect x\in\mathcal X\}\) given $\vect Z_i$:
\[
  U_i \;\indep\; \vect Z_i,
  \qquad
  U_i \;\indep\; \{Y_i^{(\vect x)}:\vect x\in\mathcal X\}\mid \vect Z_i .
\]
Then
\[
  \vect X_i \;\indep\; \{Y_i^{(\vect x)}:\vect x\in\mathcal X\}\ \bigm|\ \vect Z_i.
\]
In particular, for each fixed $\vect x\in\mathcal X$,
$Y_i^{(\vect x)}\indep \vect X_i\mid \vect Z_i$.
\end{lem}

The proof follows \citet[Lemma~1]{wang2019blessings} and the conditional representation results of Kallenberg; for completeness, we summarize the argument in Appendix~\ref{sec:Kallenberg-construction-proof}.

The remainder of the paper estimates a substitute confounder
$\hat{\vect Z}_{i,nm}$ from the causes (via a structured assignment model) and then uses this fitted score for outcome adjustment. The next proposition records a basic plug-in stability property: if $\hat{\vect Z}_{i,nm}$ consistently approximates $\vect Z_i$ in mean square and the regression functional is Lipschitz in the latent score, then replacing $\vect Z_i$ by $\hat{\vect Z}_{i,nm}$ perturbs the target expectation by at most order $\varepsilon_{n,m}$ along the joint sequence for the number of units and causes. We will denote by $\hat{\vect Z}_{i,nm}=f_{nm}(\vect X_i;\vect X_{1:n})$ a fitted
score computed by learning a common assignment map from $\vect X_{1:n}$ and then evaluating it at unit $i$. 

\begin{proposition}[Stability under estimated scores]
\label{prop:plugin-stability}
Fix a treatment regime $\vect x\in\mathcal X$ and define
\begin{align}
  r(\vect x,\vect z):=\mathbb E[Y_i\mid \vect X_i=\vect x,\vect Z_i=\vect z].
\end{align}
Assume:
\begin{enumerate}[label=(S\arabic*)]
\item \label{ass:S1} (Lipschitz continuity in $\vect z$) There exists a set
$\mathcal Z_{n,m}\subseteq\mathbb R^H$ containing both $\vect Z_i$ and
$\hat{\vect Z}_{i,nm}$ almost surely, and a constant $L(\vect x)<\infty$, such
that
\[
  |r(\vect x,\vect z)-r(\vect x,\vect z')|
  \le L(\vect x)\|\vect z-\vect z'\|
  \qquad\text{for all }\vect z,\vect z'\in\mathcal Z_{n,m}.
\]

\item \label{ass:S2} (Score accuracy) Along the joint asymptotic sequence under consideration,
\[
  \mathbb E_0\|\hat{\vect Z}_{i,nm}-\vect Z_i\|^2\le \varepsilon_{n,m}^2,
  \qquad \varepsilon_{n,m}\to0.
\]

\item \label{ass:S3} (Integrability) $\mathbb E_0|r(\vect x,\vect Z_i)|^2<\infty$.
\end{enumerate}
Then
\begin{align}
  \mathbb E_0\{r(\vect x,\hat{\vect Z}_{i,nm})-r(\vect x,\vect Z_i)\}^2
  \le L(\vect x)^2\varepsilon_{n,m}^2.
\end{align}
\end{proposition}

From Proposition \ref{prop:plugin-stability} it also immediately follows
\begin{align}
  \mathbb E_0\{r(\vect x,\hat{\vect Z}_{i,nm})\}
  \to \mathbb E_0\{r(\vect x,\vect Z_i)\}
\end{align}
and
\begin{align}
  \left|\mathbb E_0\{r(\vect x,\hat{\vect Z}_{i,nm})\}
  -\mathbb E_0\{r(\vect x,\vect Z_i)\}\right|
  \le L(\vect x)\varepsilon_{n,m}.
\end{align}

The proposition is simply a Lipschitz continuous-mapping argument.  Its role is to isolate what must be supplied by the assignment model: an $L^2$ contraction bound for the fitted scores, indexed by both the number of units and the number of causes. The set $\mathcal Z_{n,m}$ may be replaced by any set containing the true and fitted scores along the joint sequence; if containment holds only with
probability tending to one, the squared contribution from the complement must be $o(\varepsilon_{n,m}^2)$.  The score accuracy condition may follow from posterior contraction or frequentist consistency of the fitted factor model, but
is not implied by increasing the number of causes alone without additional structure and sample size conditions.

\begin{example}[Using Proposition~\ref{prop:plugin-stability}]
Suppose the latent variable is one-dimensional and
$r(\vect x,z)=x_1+z$, which is 1-Lipschitz in $z$.  If
$\mathbb E_0|\hat Z_{i,nm}-Z_i|^2\le\varepsilon_{n,m}^2$, then
\[
  \left|\mathbb E_0\{r(\vect x,\hat Z_{i,nm})\}
  -\mathbb E_0\{r(\vect x,Z_i)\}\right|\le\varepsilon_{n,m}.
\]
The rate $\varepsilon_{n,m}$ depends on the statistical problem used to learn the assignment map.  For example, it may involve both $n$ and $m$ under a high-dimensional factor-model contraction theorem.
\end{example}

For completeness, Appendix~\ref{sec:stability-under-consistent-substitute} establishes
 a stronger conditional stability statement under consistent
substitution, showing that conditioning on $\hat{\vect Z}_{i,nm}$ asymptotically
preserves the relevant conditional expectations when
$\hat{\vect Z}_{i,nm} \to \vect Z_i$.

\section{Sparse Factor Assignment Models and Substitute Confounders}
\label{sec:sparse-assignment}
Section~\ref{subsec:deconfounder} introduced the deconfounder framework: a latent variable $\vect Z_i$ renders the causes conditionally independent through an assignment model of the form \eqref{eq:assignment-model}, and an estimator
$\hat{\vect Z}_{i,nm}=f_{\hat\alpha,\hat\theta}(\vect X_i;\vect X_{1:n})$ (e.g.\ a posterior mean or mode
under the fitted assignment model $p_{\hat\theta}$) is used for outcome adjustment. In this section, we specify a concrete assignment model and explain why structural restrictions are needed for $\hat{\vect Z}_{i,nm}$ to behave as a
\emph{multi-cause} substitute confounder in high-dimensional regimes. This plan works only if the assignment model actually recovers the shared source of variation that confounds the causes. When the number of causes $m$ is large, this becomes delicate. A latent-factor model with too much freedom can
fit the joint distribution of $\vect X$ in ways that are unhelpful for causal inference and require far more samples to be estimable.

\subsection{Why very flexible assignment models can fail}
\label{subsec:why-sparsity}

Lemma~\ref{lem:kallenberg-implies-unconf} formalizes the substitute confounder premise via a Kallenberg representation: conditional on $\vect Z_i=z$, the residual variability in each cause is generated by an independent idiosyncratic draw $U_{ij}$. Thus, after conditioning on $\vect Z_i$, the remaining variation in the causes is carried by the independent $U_{ij}$. For deconfounding, it is therefore important that $X_{ij}\mid \vect Z_i=z$ remain non-degenerate for many causes and relevant values of $z$; equivalently, we want $\Var(X_{ij}\mid \vect Z_i=z)>0$ rather than a conditional law concentrated at a point.

\medskip

If the assignment model is too flexible, the estimated factor score $\hat{\vect Z}_{i,nm}=f_{\hat\alpha,\hat\theta}(\vect X_i;\vect X_{1:n})$ can encode not only shared structure but also much of the cause-specific randomness in \eqref{eq:kallenberg_con}. A strong sufficient condition for loss of latent overlap is the existence of a measurable map $g$ such that
\begin{equation}
\label{eq:suff-stat}
  \vect X_i = g(\hat{\vect Z}_{i,nm})\qquad\text{a.s.}
\end{equation}
In that case, conditioning on $\hat{\vect Z}_{i,nm}$ leaves essentially no residual
variation in $\vect X_i$: the conditional distribution
$\vect X_i\mid \hat{\vect Z}_{i,nm}$ is concentrated at a single point, so its
support collapses and latent-level positivity cannot hold on any nontrivial treatment region. A useful quantitative diagnostic for this pathology is conditional variance collapse:
\[
  \mathbb E\!\left[\|\vect X_i-\mathbb E(\vect X_i\mid \hat{\vect Z}_{i,nm})\|^2\right].
  %\approx 0.
\]
When this quantity is close to zero, $\vect X_i$ is nearly determined by
$\hat{\vect Z}_{i,nm}$, indicating that the estimated score is acting almost as a
sufficient statistic for the causes rather than as a coarse substitute
confounder.
\medskip

These are two distinct failure modes: the first is \emph{over-encoding}, where $\hat{\vect Z}_{i,nm}$ retains so much information about $\vect X_i$ that
conditioning on $\hat{\vect Z}_{i,nm}$ nearly pins down the causes; the second is \emph{under-sharing}, where latent components explain only one or a few causes and therefore fail to capture genuine multi-cause confounding structure. The
second pathology is closely related to the critique in \citet{damour2019multi}: once latent coordinates behave like single-cause or
near single-cause summaries, the multi-cause logic underlying Lemma~\ref{lem:kallenberg-implies-unconf} no longer applies cleanly. These considerations motivate \emph{regularizing} the assignment model so that the fitted
latent representation concentrates on \emph{low-rank} and \emph{multi-cause} dependence patterns,
preserving non-degeneracy of $\vect X_i\mid \hat{\vect Z}_{i,nm}$ and discouraging single-cause artifacts.

\subsection{Gaussian factor assignment model}
\label{sec:gaussian-factor-assignment}

To avoid both over-encoding and under-sharing, we estimate $\vect Z_i$ from the
high-dimensional cause vector $\vect X_i$ using \emph{structured regularisation} that promotes (i) low effective rank and (ii) factor loadings shared across multiple causes. In this section we treat the resulting score map $\hat{\vect Z}_{i,nm}=f_{\hat\alpha,\hat\theta}(\vect X_i;\vect X_{1:n})$ abstractly, and summarize its statistical accuracy through mean square contraction rates; these rates are
the key quantities controlling the substitution error incurred when $\vect Z_i$ is replaced by $\hat{\vect Z}_{i,nm}$ in regression adjustment.

We now introduce a concrete assignment model that makes the Kallenberg-type representation in Section~\ref{subsec:deconfounder} explicit. We focus on the
Gaussian factor model as a convenient baseline: it yields closed-form conditional distributions and makes transparent
how conditional independence and residual variability arise in
$\vect X_i\mid \vect Z_i$. The subsequent regularisation conditions are then imposed to encode the multi-cause requirement and to rule out single-cause latent coordinates.

Let $\vect X_i = (X_{i1},\dots,X_{im})^\top \in \mathbb R^m$ denote the $m$ causes for unit $i$. For fixed integer $H \le m$, a standard Gaussian factor model assumes that $\vect X_i$ is generated from $H$ latent factors
$\vect Z_i \in \mathbb R^H$ via
\begin{equation}\label{eq:sparse-FA}
  \vect X_i
  =
  \Lambda \vect Z_i + \boldsymbol\varepsilon_i,
  \qquad
  \vect Z_i \stackrel{\text{i.i.d.}}{\sim} \mathcal N_H(\mathbf 0, \mathbf I_H),
  \qquad
  \boldsymbol\varepsilon_i \stackrel{\text{i.i.d.}}{\sim} \mathcal N_m(\mathbf 0, \Psi),
\end{equation}

where $\Psi = \operatorname{diag}(\sigma_1^2,\dots,\sigma_m^2) \in \mathbb R^{m\times m}$. This implies that, conditional on $\vect Z_i$, the coordinates have continuous
Gaussian marginals and factorize as:
\[
  X_{ij}\mid \vect Z_i \sim \mathcal N\!\bigl((\Lambda \vect Z_i)_j,\sigma_j^2\bigr),
  \qquad
  p_\theta(\vect X_i\mid \vect Z_i)=\prod_{j=1}^m p_{\theta_j}(X_{ij}\mid \vect Z_i).
\]
Because $\Psi=\diag(\sigma_1^2,\ldots,\sigma_m^2)$ in
\eqref{eq:sparse-FA}, the idiosyncratic errors are independent across
$j$, so the coordinates $\{X_{ij}\}_{j=1}^m$ are independent given
$\vect Z_i$.  Equivalently, taking mutually independent
$U_{i1},\ldots,U_{im}\sim\mathrm U(0,1)$, independent of $\vect Z_i$,
one may write
\[
  X_{ij}=f_j(\vect Z_i,U_{ij})
  :=(\Lambda\vect Z_i)_j+\sigma_j\Phi^{-1}(U_{ij}),
  \qquad j=1,\ldots,m.
\]
This is the Kallenberg representation associated with the Gaussian factor
assignment model.  The further condition
\begin{align}
\label{eq:exogeneity-U}
  (U_{i1},\ldots,U_{im})\indep
  \{Y_i^{(\vect x)}:\vect x\in\mathcal X\}\mid \vect Z_i
\end{align}
is a causal exogeneity assumption on the full data law; it is not implied by the Gaussian factor model for $\vect X$ alone.  Under
\eqref{eq:exogeneity-U}, Lemma~\ref{lem:kallenberg-implies-unconf} gives
weak unconfoundedness conditional on $\vect Z_i$.

To encode the desired multi-cause structure in the columns of $\Lambda$, we make the following estimation-side assumption. Assumption~\ref{assump:SC} is not needed for the Kallenberg construction itself:
%. It is an estimation-side restriction: 
it solely discourages estimating latent coordinates that primarily describe just a single cause. Since factor models can be rotated without changing the likelihood, the condition is imposed after the orientation convention used for factor identification, such as an ordered-shrinkage ordering, sign convention, lower-triangular constraint, or other prespecified orientation.

\begin{assumption}[Oriented effective sparsity and shared loading]
\label{assump:SC}
Let \(\delta_m>0\) be a deterministic threshold, possibly decreasing with \(m\),
and let \(\Lambda\) be written in the orientation used by the factor model.  This
orientation is part of the model specification; without it, column-wise sparsity
is not meaningful because \(\Lambda\) and \(\Lambda Q\), for an orthogonal matrix
\(Q\), give the same covariance \(\Lambda\Lambda^\top\).

For each factor \(h=1,\dots,H\), define the \emph{effective support}
\begin{align}
  \mathcal S_h(\delta_m)
  :=
  \{j\in\{1,\dots,m\}: |\lambda_{jh}|>\delta_m\},
  \qquad
  s_h(\delta_m):=|\mathcal S_h(\delta_m)|.
\end{align}
A factor is effectively inactive if \(s_h(\delta_m)=0\).  If it is active, it
must satisfy two requirements.  First, it must load non-negligibly on at least
two causes:
\begin{align}
  s_h(\delta_m)\ge 2 .
\end{align}
Second, no single cause may dominate the factor.  There exists a constant
\(\kappa_\Lambda\in(0,1)\), not depending on \(m\), such that every active factor
satisfies
\begin{align}
  \max_{1\le j\le m}
  \frac{\lambda_{jh}^2}{\sum_{k=1}^m\lambda_{kh}^2}
  \le 1-\kappa_\Lambda .
  \label{eq:loading-leverage}
\end{align}

\end{assumption}
The first requirement rules out exactly one effectively nonzero loading.  The second rules out an almost single-cause factor, for example one very large loading and one loading just above \(\delta_m\).  Ordered-shrinkage priors may still switch off unnecessary factors by shrinking all their loadings below \(\delta_m\). We note that continuous shrinkage priors such as the Normal-Gamma, multiplicative gamma process (MGP), and cumulative shrinkage process (CSP) priors do not set loadings exactly equal to zero with positive posterior probability. Thus, the condition is stated in terms of an effective threshold \(\delta_m\) and
an oriented loading matrix.  The relevant posterior concentration claim is not that \(\lambda_{jh}=0\) exactly, but that posterior mass concentrates on an orientation in which active factors have shared, non-dominated loadings and the remaining entries are negligible at the scale \(\delta_m\).

The preceding display shows that the Gaussian factor assignment model \eqref{eq:sparse-FA} has the required Kallenberg representation: the conditional product form follows from the diagonal noise covariance, and the uniform variables can be mapped to Gaussian errors through the inverse normal CDF. If the exogeneity condition \eqref{eq:exogeneity-U} is also imposed,
Lemma~\ref{lem:kallenberg-implies-unconf} yields
$\vect X_i\indep\{Y_i^{(\vect x)}:\vect x\in\mathcal X\}\mid\vect Z_i$. The same representation extends beyond Gaussian idiosyncratic errors.  Indeed, if
$\varepsilon_{i1},\dots,\varepsilon_{im}$ are independent conditional on $\vect Z_i$
(and hence $X_{i1},\dots,X_{im}$ are independent conditional on $\vect Z_i$),
and each marginal error distribution admits a strictly increasing CDF $F_{\varepsilon_j}$,
then with independent
$U_{i1},\dots,U_{im}\sim \mathrm U(0,1)$, independent of $\vect Z_i$, one may set
\[
  \varepsilon_{ij}=F_{\varepsilon_j}^{-1}(U_{ij}),
  \qquad
  X_{ij}=(\Lambda\vect Z_i)_j+F_{\varepsilon_j}^{-1}(U_{ij})
  =: f_j(\vect Z_i,U_{ij}),
\]
which is a Kallenberg construction.

\subsection{Projected outcome restrictions and identifiable projected targets}
\label{subsec:outcome-identifiability}

The preceding sections focus on learning a substitute confounder from the joint distribution of the causes.  This learning step is useful, but it is not by itself a proof of causality.  A model for \(\vect X\) alone cannot show that an unobserved variable also removes all confounding between \(\vect X\) and the potential outcomes.  We therefore separate two questions:  what target is being identified? Under what assumptions is that target a genuine mean potential outcome? This subsection defines a \emph{projected target}.  The target uses only the
part of the treatment vector that lies outside the latent confounding subspace. It is a deliberately restricted target: it does not try to estimate causal variation along directions that are strongly entangled with the latent
confounder.  Under an additional outcome restriction stated below, the projected target equals the usual mean potential outcome \(\mu(\vect x)=\E\{Y^{(\vect x)}\}\). Without that restriction, it should be read as a working, geometry-restricted target rather than as the full intervention mean.

Under the Gaussian factor assignment model \eqref{eq:sparse-FA}, dependence
among causes is carried by the low-dimensional component
\(\Lambda \vect Z_i\in\mathbb R^m\), while the idiosyncratic noise
\(\varepsilon_i\in\mathbb R^m\) is independent across causes.  This gives a
simple decomposition of the treatment space.  The \emph{confounding subspace}
\(\mathcal C:=\operatorname{span}(\Lambda)\) contains directions in which the
conditional mean of \(\vect X\) changes with \(\vect Z\).  Its orthogonal
complement \(\mathcal C^\perp\) contains directions in which \(\vect X\) varies
only through idiosyncratic noise.  Let \(P_{\mathcal C}\) and
\(P_{\mathcal C}^\perp\) be the corresponding orthogonal projections.  Since
\(P_{\mathcal C}^\perp\Lambda=0\),
\begin{equation}
\label{eq:proj-factor-disappears}
  P_{\mathcal C}^\perp \vect X_i
  =
  P_{\mathcal C}^\perp(\Lambda \vect Z_i+\varepsilon_i)
  =
  P_{\mathcal C}^\perp\varepsilon_i,
\end{equation}
so \(\E[P_{\mathcal C}^\perp \vect X_i\mid \vect Z_i]=0\).  In words, after
projecting away the confounding subspace, the remaining treatment variation has
mean zero within levels of the latent confounder.

\begin{assumption}[Projected partially linear outcome regression]
\label{ass:proj-outcome}
There exist a measurable function \(g:\mathbb R^H\to\mathbb R\) and a vector
\(\beta\in\mathbb R^m\) such that, with \(U_i:=P_{\mathcal C}^\perp\vect X_i\),
\begin{equation}
\label{eq:proj-plm}
  \E\!\left[Y_i\mid \vect X_i,\vect Z_i\right]
  = g(\vect Z_i)+\beta^\top U_i .
\end{equation}
Moreover, the potential outcome mean depends on an intervention \(\vect x\) only
through its projected component:
\begin{equation}
\label{eq:proj-potential}
  \E\!\left[Y_i^{(\vect x)}\mid \vect Z_i\right]
  = g(\vect Z_i)+\beta^\top P_{\mathcal C}^\perp\vect x,
  \qquad \vect x\in\mathbb R^m .
\end{equation}
\end{assumption}

Let $P_0$ denote the data-generating law; convergence in $P_0$-probability is
written as $\xrightarrow{P_0}$.

\begin{assumption}[Known or consistently estimated projection]
\label{ass:projection-known}
The projection \(P_{\mathcal C}^\perp\) is either known from design or scientific
structure, or there is an estimator \(\widehat P_{\mathcal C}^\perp\), based only
on the assignment model for \(\vect X\), such that
\begin{align}
  \|\widehat P_{\mathcal C}^\perp-P_{\mathcal C}^\perp\|_{\mathrm{op}}
  \xrightarrow{P_0}0 .
  \label{eq:projection-consistency}
\end{align}
Here \(\|\cdot\|_{\mathrm{op}}\) is the operator norm, the largest amount by which
a matrix can stretch a unit vector.  If this condition is not imposed, the
projected identification result below should be read as conditional on the
chosen projection, not as a claim that the projection has been identified from
\(\vect X\).
\end{assumption}

Only the projected coefficient \(\beta_\perp:=P_{\mathcal C}^\perp\beta\) is identified.  Components of \(\beta\) inside \(\mathcal C\) do not enter \eqref{eq:proj-plm} or \eqref{eq:proj-potential}.  This is a mathematical necessity: \(U=P_{\mathcal C}^\perp\vect X\) lies in the lower-dimensional subspace \(\mathcal C^\perp\), so \(\Var(U)\) is singular as an \(m\times m\) matrix unless \(\mathcal C=\{0\}\).  We therefore use coordinates inside \(\mathcal C^\perp\), or equivalently a Moore-Penrose inverse.

\begin{proposition}[Identification of the projected target]
\label{prop:proj-identification}
Let \(\vect X\in\mathbb R^m\) denote the cause vector and let \(\vect Z\) be a
latent confounder.  Treat \(\mathcal C\subset\mathbb R^m\) as a fixed linear
subspace with known projection \(P_{\mathcal C}^\perp\).  Define
\(U:=P_{\mathcal C}^\perp\vect X\) and assume \(\E[U\mid \vect Z]=0\).  Let
\(B_{\mathcal C}\in\mathbb R^{m\times r}\), \(r=\dim(\mathcal C^\perp)\), have
orthonormal columns spanning \(\mathcal C^\perp\), and set
\(W:=B_{\mathcal C}^\top\vect X\).  Suppose that \(\Var(W)\) is nonsingular and
that Assumption~\ref{ass:proj-outcome} holds.

Then the identifiable projected coefficient is
\begin{equation}
\label{eq:gamma-proj-id}
  \gamma
  :=B_{\mathcal C}^\top\beta_\perp
  =\Var(W)^{-1}\Cov(W,Y),
  \qquad
  \beta_\perp:=B_{\mathcal C}\gamma
  =\Var(U)^+\Cov(U,Y),
\end{equation}
where \(\Var(U)^+\) denotes the Moore-Penrose inverse.  Thus \(\beta_\perp\) is
identified from the observed law of \((Y,\vect X)\) together with the fixed
projection \(P_{\mathcal C}^\perp\).

Define the projected working target
\begin{equation}
\label{eq:mu-perp-id}
  \mu_\perp(\vect x)
  :=
  \E\!\left[Y-\beta_\perp^\top U\right]
  +
  \beta_\perp^\top P_{\mathcal C}^\perp\vect x .
\end{equation}
This target is identified.  It is not, in general, the full mean potential
outcome.  Under Assumption~\ref{ass:proj-outcome}, which rules out causal
variation along \(\mathcal C\), \(\mu_\perp(\vect x)=\mu(\vect x)\).  Without that
assumption, \(\mu_\perp\) should be interpreted as a geometry-restricted target
that deliberately avoids the confounded directions.
\end{proposition}

The proof of Proposition~\ref{prop:proj-identification} is a moment calculation
on the coordinates \(W=B_{\mathcal C}^\top\vect X\) and is included in
Appendix~\ref{sec:appendix-proof-proj-id}.  If the projection is estimated, the
same formula can be used as a plug-in estimator; consistency then requires
Assumption~\ref{ass:projection-known} together with the usual continuity and
eigenvalue-separation conditions ensuring that \(\Var(W)\) remains nonsingular in
the limit.

Under the Gaussian factor assignment model \eqref{eq:sparse-FA} with
\(\mathcal C=\mathrm{span}(\Lambda)\) and \(\E[\varepsilon_i\mid\vect Z_i]=0\),
\eqref{eq:proj-factor-disappears} gives
\(U_i=P_{\mathcal C}^\perp\vect X_i=P_{\mathcal C}^\perp\varepsilon_i\), hence
\(\E[U_i\mid\vect Z_i]=0\).  Therefore Proposition~\ref{prop:proj-identification}
applies once the projection condition in Assumption~\ref{ass:projection-known}
is accepted.  Section~\ref{sec:bayesian-assignment} returns to this geometry in a Bayesian setting, where uncertainty about \(\mathcal C\) is propagated or used for regularization rather than treated as a causal fact learned from \(\vect X\) alone. Appendix~\ref{sec:appendix-outcome-identifiability} also records a soft regularized target that permits coefficients along \(\mathcal C\) but shrinks them.

\subsection{Asymptotic properties under a consistent substitute}
\label{subsec:asymptotics-sparse-deconf}

We study the consistency of regression adjustment using factor score summaries
$\hat{\vect Z}_{i,nm}$ in place of the latent confounder $\vect Z_i$. Recall
\[
  \mu(\vect x)=\E\!\left[Y_i^{(\vect x)}\right],
  \qquad
  \Delta(\vect x,\vect x')=\mu(\vect x')-\mu(\vect x),
\]
and define the outcome regression
\[
  r(\vect x,\vect z)
  := \E\!\left[ Y_i \mid \vect X_i=\vect x,\ \vect Z_i=\vect z \right].
\]

We continue to write $\Prob_0$ and $\E_0$ for probability and expectation under the data-generating law $P_0$. The projected target \(\mu_\perp(\vect x)\) in Proposition~\ref{prop:proj-identification} is a geometry-restricted estimand. It is useful when one deliberately avoids treatment directions in the confounding subspace, and it equals the full mean potential outcome \(\mu(\vect x)\) only under the outcome restriction in Assumption~\ref{ass:proj-outcome}.  In contrast, this subsection studies consistency for the full target \(\mu(\vect x)\).  We therefore impose latent-level ignorability and positivity directly; these assumptions identify \(\mu(\vect x)=\E\{r(\vect x,\vect Z_i)\}\), while the remaining conditions below control the error from replacing \(\vect Z_i\) by \(\hat{\vect Z}_{i,nm}\).

\begin{assumption}[Latent ignorability and positivity]
\label{ass:latent-ign-positivity}
There exists a latent variable $\vect Z_i$ such that:
\begin{enumerate}
  \item (Weak unconfoundedness) for all $\vect x\in\mathcal X$,
  \[
    Y_i^{(\vect x)} \;\indep\; \vect X_i \mid \vect Z_i.
  \]
  \item (Latent-level positivity) for all $(\vect x,\vect z)$ in the support of interest,
  \[
    p_{\vect X\mid \vect Z}(\vect x\mid \vect z)>0
    \quad\text{(or, for discrete $\vect X$, }\Prob(\vect X_i=\vect x\mid \vect Z_i=\vect z)>0\text{)}.
  \]
\end{enumerate}
\end{assumption}

Given an estimator $\hat r_n(\vect x,\vect z)$ of $r(\vect x,\vect z)$ and factor score summaries $\hat{\vect Z}_{i,nm}$, define
\begin{align}
  \hat\mu_n(\vect x)
  &:= \frac{1}{n} \sum_{i=1}^n \hat r_n\!\bigl(\vect x,\hat{\vect Z}_{i,nm}\bigr),
  \label{eq:mu-hat-ATE}
  \\
  \hat\Delta_n(\vect x,\vect x')
  &:= \hat\mu_n(\vect x')-\hat\mu_n(\vect x).
  \label{eq:Delta-hat-ATE}
\end{align}

\begin{thm}[Consistency of regression-adjusted mean potential outcomes]
\label{thm:mu-consistency-ATE}
Fix $\vect x\in\mathcal X$. Suppose Assumption~\ref{ass:latent-ign-positivity} holds and:
\begin{enumerate}[label=(C\arabic*), leftmargin=1.4em, itemsep=0.25em, topsep=0.25em]
\item (\emph{Lipschitz in $\vect z$}) There exists a set
$\mathcal Z_{n,m}(\vect x)\subseteq\mathbb R^H$ containing the true latent scores
$\vect Z_i$ and the fitted scores $\hat{\vect Z}_{i,nm}$, $i=1,\ldots,n$, almost
surely along the joint sequence, and a constant $L(\vect x)<\infty$ such that
\[
|r(\vect x,\vect z)-r(\vect x,\vect z')|
\le L(\vect x)\,\|\vect z-\vect z'\|
\qquad\forall\,\vect z,\vect z'\in\mathcal Z_{n,m}(\vect x).
\]
If containment in $\mathcal Z_{n,m}(\vect x)$ holds only with probability tending to one, the squared contribution from its complement is assumed to be negligible.
\item (\emph{Average score accuracy}) The fitted scores satisfy
\[
  \E_0\left[\frac1n\sum_{i=1}^n
  \|\hat{\vect Z}_{i,nm}-\vect Z_i\|^2\right]
  \le \varepsilon_{n,m}^2,
  \qquad \varepsilon_{n,m}\to0,
\]
along the joint sequence $(n,m)\to\infty$.  This expectation includes the
randomness of the first-stage fit, so the fitted scores may be dependent across
units through the common assignment model.
\item (\emph{Outcome regression calibration at the fitted scores}) The outcome
regression estimator is accurate at the score values used by the estimator:
\[
  \E_0\left[\frac1n\sum_{i=1}^n
  \bigl|\hat r_n(\vect x,\hat{\vect Z}_{i,nm})
  -r(\vect x,\hat{\vect Z}_{i,nm})\bigr|\right]
  \to0.
\]
The treatment level $\vect x$ is fixed; the randomness comes from the sample,
the fitted outcome regression, and the fitted scores.
\item (\emph{Oracle empirical average along the joint sequence}) The ideal latent
regression satisfies
\[
  \E_0\left|\frac1n\sum_{i=1}^n r(\vect x,\vect Z_i)
  -\E_0\{r(\vect x,\vect Z_i)\}\right|\to0.
\]
For independent units, a sufficient condition is
$\Var_0\{r(\vect x,\vect Z_i)\}/n\to0$; in particular this holds if $\sup_m\E_0\{r(\vect x,\vect Z_i)^2\}<\infty$ along the joint sequence.
\end{enumerate}
Then, as $(n,m)\to\infty$ with $\varepsilon_{n,m}\to0$,
\[
  \hat\mu_n(\vect x)\xrightarrow{\Prob_0}\mu(\vect x).
\]
Moreover,
\begin{equation}
\label{eq:mu-bias-ATE}
  \bigl|\E_0[\hat\mu_n(\vect x)]-\mu(\vect x)\bigr|
  \le L(\vect x)\,\varepsilon_{n,m} + o(1),
\end{equation}
where $o(1)\to 0$ along the same joint sequence $(n,m)\to\infty$.
\end{thm}

Full proof of Theorem~\ref{thm:mu-consistency-ATE} is in Appendix~\ref{sec:proof-mu-consistency}. The bound \eqref{eq:mu-bias-ATE} shows that the substitution bias induced by using
estimated scores $\hat{\vect Z}_{i,nm}$ in place of the ideal latent confounder $\vect Z_i$
is of order $\varepsilon_{n,m}$. In particular, if $\varepsilon_{n,m}=O(m^{-1/2})$ along a particular joint sequence, then the substitution component of the bias decays at that rate.

\begin{cor}[Consistency and bias of the ATE estimator]
\label{cor:ATE-consistency-ATE}
Let $\vect x,\vect x'\in\mathcal X$ and define
$\Delta(\vect x,\vect x'):=\mu(\vect x')-\mu(\vect x)$ and
$\hat\Delta_n(\vect x,\vect x'):=\hat\mu_n(\vect x')-\hat\mu_n(\vect x)$ as in
\eqref{eq:Delta-hat-ATE}. Suppose that the conditions of
Theorem~\ref{thm:mu-consistency-ATE} hold for both treatment regimes $\vect x$
and $\vect x'$. Then
\[
  \hat\Delta_n(\vect x,\vect x') \xrightarrow{\Prob_0} \Delta(\vect x,\vect x')
  \qquad (n\to\infty,\; m\to\infty,\; \varepsilon_{n,m}\to0),
\]
and the bias satisfies
\begin{equation}
\label{eq:ATE-bias-main}
  \bigl|\E_0[\hat\Delta_n(\vect x,\vect x')] - \Delta(\vect x,\vect x')\bigr|
  \le \bigl(L(\vect x)+L(\vect x')\bigr)\,\varepsilon_{n,m} + o(1),
\end{equation}
where $L(\vect x)$ and $L(\vect x')$ are the Lipschitz constants appearing in
Theorem~\ref{thm:mu-consistency-ATE}.
\end{cor}

The proof of Corollary~\ref{cor:ATE-consistency-ATE} follows immediately from  Theorem~\ref{thm:mu-consistency-ATE} applied to $\vect x$ and $\vect x'$ and is omitted for brevity.

%\paragraph{When calibration is plausible.}
Condition (C3) is intentionally stated at the score values used by the estimator, rather than uniformly over all latent values.  In a correctly specified parametric linear outcome model with bounded second moments, it follows from the consistency of the regression coefficients and an average moment bound on $(\vect X_i,\hat{\vect Z}_{i,nm})$.  More flexible nonparametric outcome regressions require corresponding prediction-risk conditions on the empirical
distribution of the fitted scores; the theorem does not assume such calibration follows automatically from the assignment model.

Theorem~\ref{thm:mu-consistency-ATE} isolates the key mechanism underlying all subsequent results: causal estimation is stable under substitution $\vect Z_i \rightsquigarrow \hat{\vect Z}_{i,nm}$ provided the substitution error contracts. In Section~\ref{sec:bayesian-assignment}, we show that Bayesian assignment models implement this same mechanism while propagating uncertainty in the confounding structure. Global-local shrinkage is the mechanism that preserves overlap in $(\vect x,\hat{\vect Z}_{nm})$. By enforcing a low-dimensional multi-cause structure and suppressing single-cause variation, sparsity prevents the factor scores from becoming near-sufficient statistics for $\vect X$. Thus positivity is required only at the latent level $(\vect X,\vect Z)$, and the mean square contraction condition in Theorem~\ref{thm:mu-consistency-ATE} ensures
that $\hat{\vect Z}_{i,nm}$ is an asymptotically valid proxy for $\vect Z_i$. This avoids the overlap collapse and other latent pathologies documented in \cite{ogburn2019comment,damour2019multi}, while yielding consistent estimates of mean potential outcomes and ATEs.

\section{Bayesian assignment models}
\label{sec:bayesian-assignment}

Section~\ref{sec:sparse-assignment} demonstrated why structured regularisation is necessary in assignment models to enable the substitute deconfounder framework in high-dimensional regimes. In this section, we propose using a Bayesian framework to estimate $\vect Z_i$ with the desired low-rank and multi-cause dependence structure. This has two primary benefits. First, we use
established shrinkage priors to concentrate posterior mass on assignment mechanisms that satisfy the structural conditions for valid substitute confounders (conditional independence, overlap, and finite-dimensional latent structure). Second, the posterior distribution for $\vect Z_i$ characterizes uncertainty in the confounding geometry itself, rather than committing to a single point estimate. In particular, the posterior over assignment parameters induces a data-adaptive metric on the space of causes, encoding both dominant confounding directions and uncertainty about them.

Throughout this section, we focus on the Gaussian factor assignment model in
\eqref{eq:sparse-FA},
\[
  \vect X_i = \Lambda \vect Z_i + \boldsymbol\varepsilon_i,
\]
with $\Lambda$ and $\Psi$ as defined in Section~\ref{sec:gaussian-factor-assignment}.
Let $\theta_X := (\Lambda,\Psi,\alpha)$ collect the \emph{assignment} parameters,
where $\alpha$ governs the latent prior for $\vect Z_i$. Let $\theta_Y$ denote the
parameters of the outcome regression model, and write $\theta:=(\theta_X,\theta_Y)$
for the full parameter set.

We write $\Pi_n(\cdot)\equiv \Pi(\cdot\mid \vect X_{1:n},Y_{1:n})$ for the joint
posterior distribution of $\theta$ under the assignment--outcome model.
When discussing the assignment mechanism and the induced confounding geometry, we
also use the assignment posterior based only on the causes,
\[
  \Pi_n^{\mathrm{assn}}(\cdot)\equiv \Pi(\theta_X\in\cdot\mid \vect X_{1:n}),
\]
i.e.\ the marginal posterior for $\theta_X$ learned from $\vect X$ alone.

To prevent outcome-to-assignment feedback, we adopt the corresponding modular
(\emph{cut}) posterior \citep{bayarri2009cut}:
\[
  \Pi_n^{\mathrm{cut}}(d\theta_X,d\theta_Y)
  :=
  \Pi_n^{\mathrm{assn}}(d\theta_X)\,
  \Pi(d\theta_Y\mid \vect X_{1:n},Y_{1:n},\theta_X),
\]
which updates the assignment parameters using $\vect X_{1:n}$ only, while the
outcome stage conditions on $\theta_X$ (equivalently, on factor score summaries
$\hat{\vect Z}_i(\theta_X)$). In this formulation, the substitute confounder is
the posterior distribution of $\vect Z_i$ given $\vect X_i$ under $\theta_X$, and
causal functionals propagate uncertainty by averaging over posterior draws of
$(\theta_X,\theta_Y,\vect Z_{1:n})$ under $\Pi_n^{\mathrm{cut}}$.

We consider two complementary classes of shrinkage priors for the loadings
$\Lambda=(\lambda_{j\ell})_{m\times H}$.

\smallskip
\noindent\emph{Global-local shrinkage.}
A canonical specification is the Normal-Gamma prior
\citep{brown2010inference,carvalho2008high}:
\[
  \lambda_{j\ell}\mid \tau_\ell,\phi_{j\ell}\sim \mathcal N\!\left(0,\;(\tau_\ell\phi_{j\ell})^{-1}\right),
  \qquad
  \phi_{j\ell}\sim \mathrm{Ga}\!\left(\nu/2,\nu/2\right),
\]
with factor-specific global precisions $\{\tau_\ell\}$ (endowed with a prior)
controlling overall rank, and local precisions $\{\phi_{j\ell}\}$ suppressing idiosyncratic (near single-cause) loadings.

\smallskip
\noindent\emph{Ordered shrinkage.}
To encourage automatic rank reduction, we also use ordered constructions such as
the multiplicative gamma process (MGP) prior \citep{bhattacharya2011sparse}:
\[
  \lambda_{j\ell}\mid \tau_\ell,\phi_{j\ell}\sim \mathcal N\!\left(0,\;(\tau_\ell\phi_{j\ell})^{-1}\right),
  \qquad
  \tau_{\ell} = \prod_{h=1}^\ell \delta_h,
  \qquad
  \delta_h \sim \mathrm{Ga}(a_h,1),
\]
or the
cumulative shrinkage process (CSP) \citep{legramanti2020bayesian}:
\begin{equation*}
\begin{gathered}
    \lambda_{j\ell}\mid \theta_\ell
    \sim \mathcal N(0,\theta_\ell),
    \qquad
    \theta_\ell \mid \pi_\ell
    \sim (1-\pi_\ell)P_0 + \pi_\ell \delta_{\theta_\infty}, \\
    \pi_\ell = \sum_{h=1}^{\ell}\omega_h, 
    \qquad
    \omega_h = v_h \prod_{m=1}^{h-1}(1-v_m), 
    \qquad
    v_h \sim \mathrm{Beta}(1,\alpha).
\end{gathered}
\end{equation*}
These impose increasing shrinkage across factor indices $\ell$ so that higher-indexed factors are progressively shrunk toward zero.  

\subsection{Posterior identification set}
\label{subsec:posterior-id-set}

Let $p_{\theta_X}(\vect X)$ denote the marginal density of the causes induced by
an assignment model at parameter $\theta_X$. The conditions below describe when an assignment model has the structural form needed for substitute confounder adjustment. They are assignment-side conditions only: by themselves they do not
establish causal ignorability, because ignorability is a property of the full law
of $(\vect X,\{Y^{(\vect x)}\}_{\vect x\in\mathcal X})$, not of the marginal law
of $\vect X$.

Recall from Section~\ref{subsec:deconfounder} that only a restricted class of assignment mechanisms lead to valid substitute confounders.

\begin{defn}[Identification-compatible assignment mechanisms]
\label{def:posterior-ID}
An assignment model $p_{\theta_X}(\vect X)$ is \emph{assignment-compatible} if
there exists a latent variable $\vect Z_{\theta_X}$ such that:

\smallskip
\noindent\textit{Conditional independence.}
The causes are conditionally independent given $\vect Z_{\theta_X}$:
\begin{align}
  p_{\theta_X}(\vect X \mid \vect Z_{\theta_X})
  =
  \prod_{j=1}^m p_{\theta_X,j}(X_j \mid \vect Z_{\theta_X}).
\end{align}

\smallskip
\noindent\textit{Quantitative latent overlap.}
There is a constant \(\kappa_X>0\) and a set of treatment directions \(\mathcal D\)
that contains the directions used by the causal estimand such that, for
\(\Prob_{\vect Z_{\theta_X}}\)-almost every \(\vect z\),
\begin{align}
  \Var_{\theta_X}(a^\top\vect X\mid \vect Z_{\theta_X}=\vect z)
  \ge \kappa_X
  \qquad\text{for every } a\in\mathcal D,\ \|a\|=1 .
\end{align}
For the Gaussian factor model with diagonal \(\Psi\), a simple sufficient condition is \(\lambda_{\min}(\Psi)\ge \kappa_X\).  This quantitative condition rules out conditional variances that are positive only in a formal sense but too small to support stable comparisons. It is a stability margin strengthening the qualitative non-degeneracy condition discussed in Section \ref{subsec:why-sparsity}.

\smallskip
\noindent\textit{Finite latent complexity.}
The latent variable is finite-dimensional:
\begin{align}
  \dim(\vect Z_{\theta_X})<\infty .
\end{align}

\smallskip
\noindent\textit{Effective multi-cause structure in factor models.}
When $p_{\theta_X}$ is represented by the Gaussian factor model
\eqref{eq:sparse-FA}, its loading matrix satisfies the effective support
condition in Assumption~\ref{assump:SC} for the threshold sequence used in the
analysis.

\smallskip
\noindent
We denote the collection of all such assignment models by $\mathcal M_{\mathrm{A}}$.
\end{defn}

Although $\mathcal M_{\mathrm{A}}$ is defined as a set of assignment \emph{models}
$p_{\theta_X}$, $\Pi_n^{\mathrm{assn}}$ is a posterior on parameters $\theta_X$.
We therefore interpret
\[
  \Pi_n^{\mathrm{assn}}(\mathcal M_{\mathrm{A}}^c)
  :=
  \Pi_n^{\mathrm{assn}}\!\bigl(\{\theta_X:
  p_{\theta_X}\notin\mathcal M_{\mathrm{A}}\}\bigr).
\]

Definition~\ref{def:posterior-ID} reframes the substitute confounder assignment requirements as a target for posterior concentration. Because priors on $\theta_X$ are typically continuous, $\Pi_n^{\mathrm{assn}}$ assigns positive mass to non-compatible mechanisms at finite $n$. The relevant assignment-side question is whether
\[
  \Pi_n^{\mathrm{assn}}(\mathcal M_{\mathrm{A}}^c)\longrightarrow 0
\]
fast enough that overlap-collapsed, near-invertible, or effectively single-cause mechanisms have negligible influence on downstream regularized causal functionals. This posterior concentration is still not a test of ignorability;
it only says that the fitted assignment model has the coarse multi-cause form
needed for the subsequent causal assumptions to be meaningful.

\begin{proposition}[Posterior-induced regularization geometry]
\label{prop:posterior-geometry}
Let $\Pi^{\mathrm{assn}}_n$ be the assignment posterior and define
\[
  H_n:=\mathbb E_{\Pi^{\mathrm{assn}}_n}(\Lambda\Lambda^\top)
  \in\mathbb R^{m\times m}.
\]
Then $H_n$ is positive semidefinite and, for every treatment direction
$a\in\mathbb R^m$,
\[
  a^\top H_n a
  =\mathbb E_{\Pi^{\mathrm{assn}}_n}
  \|\Lambda^\top a\|^2.
\]
Thus $a^\top H_n a$ is the posterior mean amount of factor-explained
variation in direction $a$.  If $H_n=U\operatorname{diag}(d_1,\ldots,d_m)U^\top$
with $d_1\ge\cdots\ge d_m\ge0$, then the columns of $U$ with large $d_j$
are precisely the directions that receive large posterior factor-explained
variance.  For $\lambda>0$, the soft residual operator
\[
  P^\perp_{H_n,\lambda}=\lambda(H_n+\lambda I_m)^{-1}
\]
has eigenvalues $\lambda/(d_j+\lambda)$ in those same directions, and
therefore downweights directions in proportion to their posterior
factor-explained variation.  This gives a direction-specific regularisation
metric learned from the assignment model.  Causal interpretation of the
resulting directions still requires the identification assumptions stated
separately; the proposition only describes the geometry induced by the
posterior assignment model.
\end{proposition}

If the factor model identifies the confounding subspace (i.e. after fixing the rotation and sign ambiguity and under a condition that separates the relevant eigenvalues), then
$H_n$ may be viewed as an estimator or posterior summary of that subspace. Without such assumptions, $H_n$ is still useful as a model-induced regularization geometry, but it is not by itself a point-identified causal object.  The projected result in Section~\ref{subsec:outcome-identifiability} therefore requires Assumption~\ref{ass:projection-known}: the projection must be known from outside the factor model or consistently estimated under additional subspace-identification conditions.

In the Bayesian assignment model, $\Lambda$ is random under
$\Pi_n^{\mathrm{assn}}$, so $\mathcal C(\theta_X):=\mathrm{span}(\Lambda)$ is a random subspace. Rather than committing to a single geometry, we either (i) average downstream causal functionals over posterior draws of the assignment parameters $\theta_X$, or (ii) use $H_n$ as a stable summary for
regularization. In particular, $H_n$ plays the role of a soft analogue of $P_{\mathcal C}$: directions with large posterior-confounding eigenvalues are persistently explained by shared latent structure, while directions with small eigenvalues behave as effectively idiosyncratic. For example, one may form a soft
projection
\[
  P^\perp_{H_n,\lambda}
  := I_m-H_n(H_n+\lambda I_m)^{-1}
  = \lambda(H_n+\lambda I_m)^{-1},
\]
which down-weights directions with large posterior confounding while retaining directions with small posterior confounding.

In practice, $H_n$ enters causal estimation either implicitly, through posterior averaging over factor scores, or explicitly, through geometry-aware outcome regularisation as in Section~\ref{subsec:outcome-identifiability}. Augmented balancing weights \citep{bruns2025augmented} can be interpreted as
enforcing balance with respect to a user-specified metric $\mathsf H$; here, Bayesian assignment models replace ad hoc or cross-validated choices of $\mathsf H$ with a posterior expectation over confounding directions learned from the joint distribution of the causes, subject to the caveat that causal interpretation requires the additional assumptions stated above.

\subsection{Bayesian causal functional consistency}
\label{subsec:bayesian-consistency}

We conclude by establishing consistency of Bayesian causal functionals under the
sparse assignment model. This result does not introduce new identification
assumptions; rather, it shows that the Bayesian procedure implements the same
stability principles as the frequentist analysis, while propagating uncertainty
about the confounding structure. We focus on the mean potential outcome $\mu(\vect x):=\E[Y_i^{(\vect x)}]$.
Recall that $\theta=(\theta_X,\theta_Y)$.  Let $\widetilde\Pi_n$ denote the posterior distribution actually used for causal estimation.  It may be the ordinary joint posterior $\Pi(\cdot\mid \vect X_{1:n},Y_{1:n})$, or it may be the modular, or \emph{cut}, posterior
\[
  \Pi_n^{\mathrm{cut}}(d\theta_X,d\theta_Y)
  =\Pi_n^{\mathrm{assn}}(d\theta_X)\,
    \Pi(d\theta_Y\mid \vect X_{1:n},Y_{1:n},\theta_X),
\]
where $\Pi_n^{\mathrm{assn}}(d\theta_X)=\Pi(d\theta_X\mid \vect X_{1:n})$.  The cut posterior learns the assignment model from the causes alone, so the outcome cannot feed back into the learned confounding geometry.  In either case, $\widetilde\Pi_{n,X}$ denotes the $\theta_X$-marginal of $\widetilde\Pi_n$ and must satisfy the assignment-stage conditions below.  For each $i$, let $\hat{\vect Z}_i(\theta_X)$ denote a measurable factor score summary under the assignment model, for example $\E_{\theta_X}[\vect Z_i\mid \vect X_i]$.

\smallskip
\noindent
The Bayesian regression-adjusted estimand is the posterior-averaged functional
\begin{equation}
\label{eq:bayes-mu}
  \mu_n^{\mathrm{pp}}(\vect x)
  :=
  \int
    \Biggl\{
      \frac{1}{n}\sum_{i=1}^n r_{\theta_Y}\!\bigl(\vect x,\hat{\vect Z}_i(\theta_X)\bigr)
    \Biggr\}\,
  \widetilde\Pi_n(d\theta),
\end{equation}
where
\[
  r_{\theta_Y}(\vect x,\vect z)
  :=
  \E_{\theta_Y}\!\left[Y_i \mid \vect X_i=\vect x,\vect Z_i=\vect z\right].
\]
Eq.~\eqref{eq:bayes-mu} is the Bayesian analogue of the regression-adjusted
estimators studied in Section~\ref{subsec:asymptotics-sparse-deconf}: instead of
conditioning on a single estimated substitute confounder, the Bayesian procedure
averages over posterior uncertainty in both the assignment mechanism (through
$\hat{\vect Z}_i(\theta_X)$) and the outcome regression (through $r_{\theta_Y}$).

\begin{thm}[Bayesian consistency of mean potential outcomes]
\label{thm:bayes-mu-consistency}
Fix \(\vect x\in\mathcal X\), and let
\(r_0(\vect x,\vect z):=\E_0[Y_i\mid \vect X_i=\vect x,\vect Z_i=\vect z]\).
Suppose the following conditions hold.
\begin{enumerate}[leftmargin=1.4em]
\item \textit{Causal identification at the latent level.}
      The full data-generating law satisfies consistency, latent ignorability
      and latent positivity as in Assumption~\ref{ass:latent-ign-positivity}, so
      that
      \[
        \mu(\vect x)=\E_0\{r_0(\vect x,\vect Z_i)\}.
      \]
      %Here latent ignorability means that, after conditioning on the true latent confounder, the observed treatment vector carries no further information about the potential outcome under \(\vect x\). We also assume the oracle empirical-average condition
      Moreover,
      \[
        \E_0\left|\frac1n\sum_{i=1}^n r_0(\vect x,\vect Z_i)
        -\E_0\{r_0(\vect x,\vect Z_i)\}\right|\to0
      \]
      along the joint sequence. %; a sufficient condition for independent units is \(\Var_0\{r_0(\vect x,\vect Z_i)\}/n\to0\), for example \(\sup_m\E_0\{r_0(\vect x,\vect Z_i)^2\}<\infty\).
\item \textit{Assignment compatibility and score contraction.}
      The true assignment mechanism belongs to \(\mathcal M_{\mathrm{A}}\), and
      the \(\theta_X\)-marginal \(\widetilde\Pi_{n,X}\) satisfies
      \(\widetilde\Pi_{n,X}(\mathcal M_{\mathrm{A}}^c)\to0\) in \(P_0\)-probability.
      Moreover, the posterior factor score summaries satisfy
      \[
        \E_0\!\left[
        \int \frac1n\sum_{i=1}^n
        \|\hat{\vect Z}_i(\theta_X)-\vect Z_i\|^2\,
        \widetilde\Pi_n(d\theta)
        \right]
        \le \varepsilon_{n,m}^2+o(1),
        \qquad \varepsilon_{n,m}\to0 .
      \]

\item \textit{Outcome posterior calibration.}
      The posterior outcome regression is calibrated at \(\vect x\) in the sense
      that
      \[
        \E_0\!\left[
        \int \frac1n\sum_{i=1}^n
        |r_{\theta_Y}(\vect x,\vect Z_i)-r_0(\vect x,\vect Z_i)|\,
        \widetilde\Pi_n(d\theta)
        \right]
        \longrightarrow 0.
      \]

\item \textit{Posterior Lipschitz and envelope conditions.}
      There are sets \(\Theta_n(\vect x)\) with
      \(\widetilde\Pi_n\{\Theta_n(\vect x)^c\}\to0\) in \(P_0\)-probability and a
      constant \(L(\vect x)<\infty\) such that, for all
      \(\theta\in\Theta_n(\vect x)\),
      \[
        |r_{\theta_Y}(\vect x,\vect z)-r_{\theta_Y}(\vect x,\vect z')|
        \le L(\vect x)\|\vect z-\vect z'\| .
      \]

      There exist nonnegative random variables \(M_{i,n}(\vect x)\), independent of the posterior draw conditional on the data, such that uniformly over
      all \(\theta\in\operatorname{supp}(\widetilde\Pi_n)\),
      \[
        |r_{\theta_Y}(\vect x,\vect Z_i)|
        + |r_{\theta_Y}(\vect x,\hat{\vect Z}_i(\theta_X))|
        \le M_{i,n}(\vect x),
        \qquad \sup_{i,n} \E_0\{M_{i,n}(\vect x)^2\}<\infty .
      \]

\end{enumerate}
Then
\[
  \mu_n^{\mathrm{pp}}(\vect x)
  \xrightarrow{P_0}
  \mu(\vect x)
  \qquad
  \text{as } n,m\to\infty,
\]
and the asymptotic bias satisfies
\[
  \bigl|
    \E_0[\mu_n^{\mathrm{pp}}(\vect x)]
    -\mu(\vect x)
  \bigr|
  \le
  L(\vect x)\varepsilon_{n,m}+o(1).
\] 
\end{thm}
Theorem~\ref{thm:bayes-mu-consistency} shows that Bayesian causal inference inherits consistency from the same stability argument used in the frequentist analysis, but only after the causal and statistical requirements are separated. Latent ignorability and positivity identify the full mean potential outcome target; assignment compatibility
and shrinkage support a coarse, overlap-preserving geometry; posterior calibration controls outcome regression error; and Lipschitz continuity transfers
factor score contraction into vanishing causal bias via Proposition~\ref{prop:plugin-stability}. For independent units, the oracle empirical-average condition follows, for example, from \(\Var_0\{r_0(\vect x,\vect Z_i)\}/n\to0\). If the factor model is identifiable only up to rotation, the score-contraction condition is interpreted after the
chosen orientation constraint, or with an infimum over admissible rotations. The envelope condition is used only to control posterior mass outside the posterior-typical Lipschitz set. A proof is given in Appendix~\ref{sec:appendix-bayes-consistency}. Consistency of posterior summaries for average treatment effects follows immediately.

\section{Experiments}
\label{sec:experiments}
In this section, we present experiments which evaluate three linked aspects of the proposed framework. First, we design a synthetic grid to illustrate the stability mechanism in
Theorems~\ref{thm:mu-consistency-ATE} and
\ref{thm:bayes-mu-consistency}: replacing the latent confounder \(Z_i\) by a
learned score \(\widehat Z_i\) induces small error in mean potential outcome functionals only when the score error is small, and this conclusion still requires the relevant outcome-side assumptions. For this reason we report both assignment diagnostics and direct errors for focal mean potential outcome
contrasts, rather than only loading recovery.

Second, we present an adaptive-shrinkage experiment to ask how the learned confounding geometry should enter regularized causal estimation. It connects the sparse FA assignment model to augmented balancing weights by comparing isotropic ridge
regularization with an FA-guided direction-specific penalty, and by separating predictive tuning from tuning criteria aligned with causal error. Third, we use the
ADNI example as a real data stress test: known baseline confounders are withheld from the factor model, allowing us to examine whether sparse substitute confounders produce empirically more stable adjustment and move the estimated
mean potential outcome contrast in the direction expected from adjustment for the withheld covariates.

\subsection{Synthetic grid study}
\label{subsec:synthetic-grid-results}

We use a common synthetic family with a scalar latent confounder
\(Z_i\sim\mathcal N(0,1)\) and a high-dimensional cause vector
\(\vect A_i\in\mathbb R^m\). The confounder loads only on a block
\(S\subset\{1,\ldots,m\}\) containing the focal causes \((A_{i1},A_{i2})\),
whereas nuisance factors load on \(S^c\):
\begin{align}
  \vect A_i
  =
  Z_i\vect w^\top
  + U_{i0}\vect v_0^\top
  + \sum_{\ell=1}^{q} U_{i\ell}\vect v_\ell^\top
  + \vect E_i .
  \label{eq:synthetic-grid-assignment}
\end{align}
Here \(\operatorname{supp}(\vect w)=S\), while
\(\operatorname{supp}(\vect v_\ell)\subseteq S^c\). The first nuisance loading \(\vect v_0\) is a high-variance ``monster'' direction and the remaining \(\vect v_\ell\)'s are dense off-block nuisance directions.  The main grid varies the nominal signal ratio \(\|\vect w\|^2/\|\vect v_0\|^2\); the appendix varies the block size \(|S|/m\) and the number of nuisance directions \(q\). The default run uses \(n=220\), \(m=200\), three independent assignment replications, and 120 outcome regenerations conditional on each fitted score. The baseline outcome model is
\begin{align}
  Y_i
  =
  \beta_1 A_{i1}+\beta_2 A_{i2}+\gamma Z_i+\varepsilon_i,
  \qquad
  \varepsilon_i\sim\mathcal N(0,\sigma_Y^2),
  \label{eq:synthetic-grid-outcome}
\end{align}
with \((\beta_1,\beta_2)=(1,1)\). This is the setting in which regression adjustment on the true \(Z_i\) is correctly calibrated. We also consider two stress tests. The first adds an interaction \(0.8 A_{i1}Z_i\), making the linear outcome adjustment misspecified while preserving the shared latent confounder. The second introduces an additional single-cause latent variable that affects both \(A_{i1}\) and \(Y_i\), violating the multi-cause substitute confounder premise.

For each assignment model, we fit a one-factor representation using the causes only and form a plug-in score \(\widehat Z_i\). The Bayesian factor models use the posterior-mean score rather than draw-by-draw posterior averaging of the causal functional, so the experiment should be read as a finite-sample plug-in illustration of the stability results. We compare a naive estimator with no substitute confounder, an oracle adjustment using the true \(Z_i\), dense Gaussian factor analysis, the Normal-Gamma (NG) sparse factor model, and the MGP
and CSP ordered-shrinkage variants.

Performance is summarized at both the assignment and causal functional levels. For assignment recovery, we report the best absolute correlation between \(Z_i\) and the fitted score, after resolving the factor sign, together with \(\mathrm{ResidCorr}(A_1,A_2\mid\widehat Z)\), the residual correlation of the two focal causes after linearly adjusting for the fitted score. For downstream estimation, let \(\widehat\beta(\widehat Z)\) denote the coefficients of \((A_1,A_2)\) in the regression of \(Y\) on \((1,A_1,A_2,\widehat Z)\), omitting \(\widehat Z\) for the naive estimator and replacing it by \(Z\) for the oracle estimator. For the contrast set
\[
  \mathcal D=\{(1,0),(0,1),(1,1)\},
\]
we define \(\widehat\Delta(d)=d^\top\widehat\beta(\widehat Z)\) and compare it
with the corresponding mean potential outcome contrast \(\Delta(d)\) under the known data-generating law. The reported functional MSE is
\[
  \mathrm{FMSE}
  =
  |\mathcal D|^{-1}\sum_{d\in\mathcal D}
  \E_{\mathrm{MC}}\!\left[
    \left\{\widehat\Delta(d)-\Delta(d)\right\}^2
  \right],
\]
where the expectation is approximated by regenerating outcomes conditional on
the fitted assignment scores. In the baseline linear design, these contrasts
are focal ATEs. We also report the coefficient-vector MSE to maintain
comparability with the earlier fixed-scenario experiments.

Figure~\ref{fig:synthetic-grid-summary} summarizes how score recovery,
residual focal dependence, and functional MSE vary with the nominal
confounder-to-nuisance signal ratio and with the outcome structure. At nominal
signal ratios \(0.01\) and \(0.02\), no fitted score closely approximates the
latent confounder and the functional MSE remains close to the naive estimator;
the oracle adjustment remains near zero, showing that the difficulty is
estimation of the substitute rather than identification conditional on the true
latent variable. Around nominal signal ratio \(0.05\), Bayes-NG and Bayes-MGP
recover scores almost perfectly aligned with \(Z\), reduce the residual focal
correlation to near zero, and attain functional MSE close to the oracle
benchmark. Dense FA is less stable in this transition regime: it succeeds in
some replications but remains susceptible to nuisance directions in others. At
stronger signal levels dense FA also catches up, while CSP continues to perform comparably to the naive estimator. Thus the main advantage of the
sparse priors is not that dense FA is uniformly inconsistent, but that
shrinkage expands the finite-sample regime in which the localized confounding
direction is recovered well enough for stable causal adjustment.

\begin{figure}[t!]
    \centering
    \includegraphics[width=\linewidth]{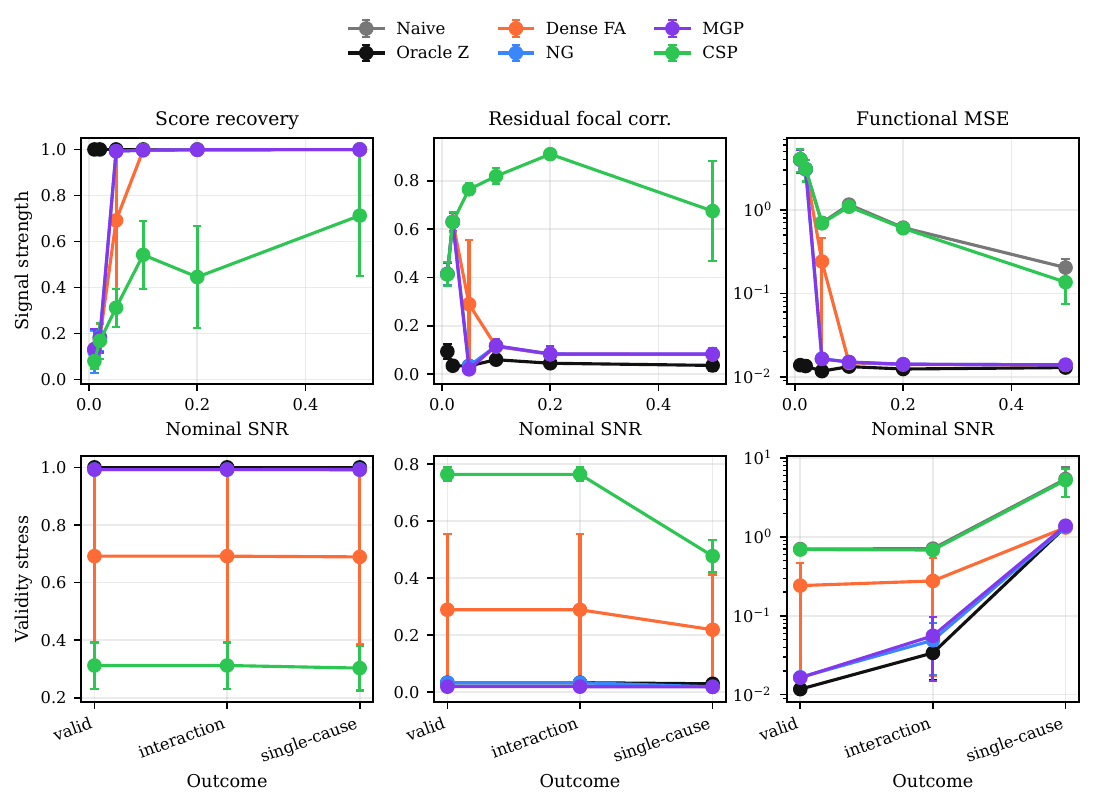}
    \caption{Synthetic grid summary. Top row: increasing nominal strength of the localised confounder relative to off-block nuisance variation. Bottom row: outcome-validity stress tests at the transition setting. The columns
    report score recovery, residual dependence between focal causes after adjustment, and MSE for focal mean-potential/ATE contrasts. Points are averages over three assignment replications; vertical bars are standard
    errors.}
    \label{fig:synthetic-grid-summary}
\end{figure}

Table~\ref{tab:synthetic-grid-main} focuses on the transition setting
(\(|S|/m=0.1\), nominal signal ratio \(0.05\), \(q=20\)). Under the valid outcome
model, NG and MGP nearly match Oracle-\(Z\): their functional MSEs are \(0.017\)
and \(0.016\), compared with \(0.012\) for the oracle and \(0.701\) for the
naive estimator. Dense FA improves over the naive estimator on average, but its
large standard error reflects replicate-level instability. CSP does not recover
the confounder under this calibration.

\begin{table}[t!]
\centering
\scriptsize
\setlength{\tabcolsep}{3.2pt}
\caption{Synthetic grid at the transition setting. Entries are means with standard errors over three assignment replications. ``Func. MSE'' is the mean-potential/ATE contrast MSE under the valid outcome model; the last two columns report the same functional MSE under outcome model and single cause
violations.}
\label{tab:synthetic-grid-main}
\begin{tabular}{lcccccc}
\toprule
Method & Score corr. & Resid. corr. & Func. MSE & Coef. MSE &
Func. MSE int. & Func. MSE single \\
\midrule
CSP & 0.312 (0.083) & 0.765 (0.024) & 0.691 (0.025) & 2.047 (0.073) & 0.682 (0.070) & 5.246 (2.076) \\
MGP & 0.992 (0.001) & 0.019 (0.013) & 0.016 (0.001) & 0.028 (0.003) & 0.056 (0.041) & 1.386 (0.067) \\
NG & 0.992 (0.002) & 0.033 (0.018) & 0.017 (0.002) & 0.027 (0.005) & 0.049 (0.031) & 1.345 (0.022) \\
Dense FA & 0.692 (0.300) & 0.289 (0.265) & 0.241 (0.224) & 0.701 (0.675) & 0.276 (0.259) & 1.314 (0.065) \\
Naive & -- & -- & 0.701 (0.032) & 2.077 (0.091) & 0.711 (0.078) & 5.418 (2.195) \\
Oracle \(Z\) & 1.000 (0.000) & 0.033 (0.005) & 0.012 (0.000) & 0.018 (0.001) & 0.034 (0.018) & 1.346 (0.064) \\
\bottomrule
\end{tabular}
\end{table}

The outcome stress tests separate assignment recovery from causal validity. The interaction setting leaves score recovery unchanged but increases functional MSE, reflecting outcome regression misspecification. The single cause violation is more severe: even Oracle-\(Z\) has functional MSE around \(1.35\). Thus the experiment does not suggest that sparse FA repairs violations of latent ignorability. Rather, it supports the logic of Sections~\ref{subsec:asymptotics-sparse-deconf}
and~\ref{subsec:bayesian-consistency}: when the outcome assumptions are valid, functional error tracks score contraction; when those assumptions are violated, excellent factor score recovery is not sufficient.

Additional diagnostic plots, the block-localisation and nuisance-rank sweeps, and the extreme scenarios figures are given in Appendix~\ref{sec:appendix-synthetic-grid}.

\subsection{Adaptive shrinkage via FA-guided direction-specific penalty}
\label{subsec:abw-adaptive}

Regularized causal estimation allows for stabilizing high-dimensional estimators: one can infer penalty parameters that determine how estimation trades bias, variance, and balance across directions of the treatment space. ABW \citep{bruns2025augmented} formalizes this idea by placing a metric on treatment imbalance; in the linear regression form used here the
corresponding ridge estimator solves
\[
  (\boldsymbol A^\top \boldsymbol A+nH)\hat\theta=\boldsymbol A^\top Y .
\]
and an isotropic ABW ridge penalty can be interpreted via \(H=\lambda I_p\).  Our question is
whether the assignment-stage FA geometry can supply a more targeted penalty,
\[
  H_{\rm FA}(\lambda_0,\lambda_1)
  =\lambda_0 I_p+\lambda_1\widehat P_{\rm conf},
\]
where \(\widehat P_{\rm conf}\) is the plug-in orthogonal projector onto the
span of the posterior mean loading vector.  This is a computationally simple
hard-projector version of the posterior geometry in
Proposition~\ref{prop:posterior-geometry}; projectors are not averaged over MCMC
draws in this experiment.

The assignment design is the block-sparse model used above,
\[
  \boldsymbol A
  =
  Z\boldsymbol w_{\rm conf}^\top
  +\boldsymbol U_{\rm nuis}\boldsymbol B_{\rm nuis}^\top
  +\boldsymbol\varepsilon_A,
  \qquad
  Y=\boldsymbol A\theta^\star+h(Z)+\varepsilon_Y ,
\]
with \(n=400\), \(p=80\), \(p_{\rm conf}=30\), two nuisance factors, and loading
scales \(\sigma_{\rm conf}=8\) and \(\sigma_{\rm nuis}=3\).  The target
coefficient is
\[
  \theta^\star=\epsilon\,\theta_F+(1-\epsilon)\theta_U,
  \qquad
  \theta_F\in{\rm span}(\boldsymbol w_{\rm conf}),\quad
  \theta_U\perp{\rm span}(\boldsymbol w_{\rm conf}),
\]
so \(\epsilon\) varies the overlap between the causal target and the confounded
direction.  The grid uses
\(\epsilon\in\{0,0.05,0.1,0.2,0.3,0.4,0.6,0.8\}\).  For each Monte Carlo
replicate we fit a one-factor Normal-Gamma FA model using the causes only.
Four chains are run and the plug-in projector is selected by the Gaussian FA
covariance score for \(\boldsymbol A\); this assignment-stage selection uses neither \(Y\), \(\theta^\star\), nor the true loading vector.  The selected
chains recover the simulated confounding direction well, with mean absolute cosine \(0.998\) (standard error \(0.001\)).

The comparison is organized around two questions:  (i) conditional on a causally aligned tuning criterion, how much is gained by using the learned confounding geometry rather than an isotropic ridge penalty?  (ii) what goes wrong if the same regularization family is tuned as a prediction problem? This mirrors the practical tuning issues emphasized by \citet{bruns2025augmented}: cross-validation based choices are natural in regularized regression implementations, but they need not target the causal functional whose error we report here. We therefore report ordinary least squares (OLS), isotropic ABW ridge tuned by an oracle causal criterion, the FA-guided direction-specific penalty tuned by an oracle causal criterion, and a true-projector directional oracle.  We also report predictive cross-validation versions of isotropic ridge and FA-guided ridge to show what happens under the naive criterion of held-out outcome prediction. Oracle methods use \(\theta^\star\) and \(\Sigma_A={\rm Cov}(\boldsymbol A)\) only for simulation benchmarking. All methods are evaluated by the population causal error
\[
  (\hat\theta-\theta^\star)^\top
  \Sigma_A
  (\hat\theta-\theta^\star).
\]
The true-projector reference is not a finite-sample lower bound for every regularized estimator: a slightly rotated plug-in projector can occasionally reduce the finite-sample oracle risk.  It is included to locate the low causal error region associated with directional shrinkage.

We display the performance of the compared regularization strategies for different values of $\epsilon$ in
Figure~\ref{fig:abw-grid-summary} and Table~\ref{tab:abw-summary}. Results suggest that directional shrinkage with a causally aligned tuning criterion has low causal MSE throughout the grid. As \(\epsilon\) increases, the target itself moves toward the confounded direction and the advantage of directional shrinkage over isotropic shrinkage narrows, but it remains large over the reported range.

The predictive-CV rows in Table~\ref{tab:abw-summary} and the center and right panels of Figure~\ref{fig:abw-grid-summary} illustrate the collapse of the CV-guided regularization penalty. Because the latent confounder affects both \(\boldsymbol A\) and \(Y\), the confounded direction is
genuinely predictive of the observed outcome.  A validation criterion based on held-out prediction is therefore rewarded for preserving that direction, whereas the causal objective is improved by shrinking it.  In the grid, predictive CV selects little or no directional penalty and its causal MSE
remains close to OLS.  The right panel of Figure~\ref{fig:abw-grid-summary} reports this directly: the causal error of
predictive-CV tuning is several-fold larger than the matched oracle for isotropic ridge and hundreds of times larger for the FA-guided penalty. The conclusion is not that predictive CV is a poor predictive tuning rule; it is that prediction risk is not a valid surrogate for the bias-variance tradeoff of
the causal functional when confounded directions are predictive.  A causally aligned tuning rule must penalize directions according to their contribution to confounding and target error, even when those directions improve prediction of
the observed outcome.

The grid summaries are complemented by a direct view of the regularization surface.  Figure~\ref{fig:abw-summary} fixes two representative target alignments, \(\epsilon=0.1\) and \(\epsilon=0.3\), and plots the causal MSE over the global and directional penalties.  The low-error region lies at nonzero
directional shrinkage, showing that the gain from FA-guided regularization is not a marginal artifact of the tuning grid but a feature of the causal-error surface.  The isotropic ABW oracle is restricted to the boundary \(\lambda_1=0\), while predictive CV selects a boundary or near-boundary point with high causal error.  Thus, the figure gives a local geometric explanation for the grid-level pattern: sparse FA can recover a useful confounding direction, but realizing its causal value requires tuning aligned with the causal estimand rather than with outcome prediction alone.

\begin{table}[!htbp]
\centering
\small
\setlength{\tabcolsep}{4pt}
\caption{Adaptive-shrinkage grid: causal MSE for selected values of
\(\epsilon\). Entries are Monte Carlo means with standard errors in
parentheses. Oracle rows are simulation benchmarks.}
\label{tab:abw-summary}
\begin{tabular}{lcccc}
\toprule
& \multicolumn{4}{c}{Target-confounded mixing \(\epsilon\)} \\
\cmidrule(lr){2-5}
Method & 0 & 0.1 & 0.3 & 0.6 \\
\midrule
True \(P_{\rm conf}\) oracle & 0.806 (0.221) & 0.705 (0.241) & 0.628 (0.232) & 0.510 (0.178) \\
FA-guided oracle & 0.366 (0.068) & 0.155 (0.019) & 0.089 (0.007) & 0.055 (0.009) \\
\midrule
ABW ridge oracle & 7.193 (1.386) & 4.573 (0.943) & 2.252 (0.390) & 0.813 (0.226) \\
\midrule
FA-guided predictive CV & 35.485 (0.415) & 35.479 (0.386) & 35.391 (0.373) & 35.186 (0.403) \\
ABW ridge predictive CV & 35.461 (0.397) & 35.451 (0.395) & 35.434 (0.390) & 35.316 (0.370) \\
OLS & 35.490 (0.376) & 35.490 (0.376) & 35.490 (0.376) & 35.490 (0.376) \\
\bottomrule
\end{tabular}
\end{table}

\begin{figure}[t!]
    \centering
    \includegraphics[width=0.95\linewidth]{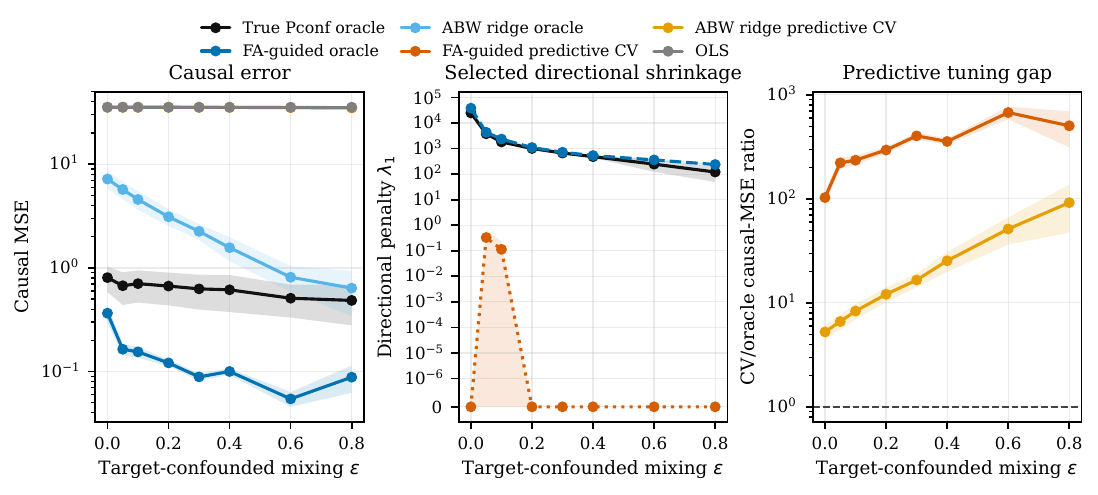}
    \caption{\emph{Adaptive-shrinkage grid.} Left: causal MSE as the target
    moves toward the confounded direction. Middle: selected directional
    penalty. Right: ratio between predictive-CV causal error and the matched
    oracle causal error, reported for the two predictive-CV methods.}
    \label{fig:abw-grid-summary}
\end{figure}

\begin{figure}[t!]
    \centering
    \includegraphics[width=0.95\linewidth]{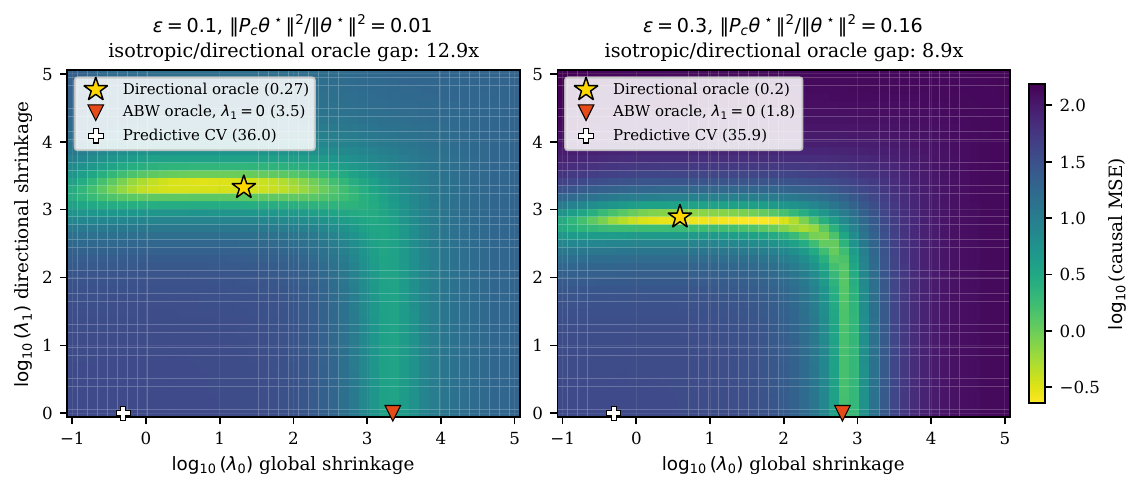}
    \caption{\emph{Causal-MSE landscapes for two reference points in the
    adaptive-shrinkage grid.} The heat maps evaluate causal MSE over global
    shrinkage \(\lambda_0\) and directional shrinkage \(\lambda_1\), using the
    true confounding projector to isolate the regularization geometry. Stars
    mark the directional oracle, triangles the isotropic ABW oracle on the
    boundary \(\lambda_1=0\), and crosses the predictive-CV selection.}
    \label{fig:abw-summary}
\end{figure}

\subsection{Alzheimer's Disease Neuroimaging Initiative}
\label{subsec:adni-example}

Next, we consider baseline data from the Alzheimer's Disease Neuroimaging
Initiative (ADNI): a multi-site study designed to combine clinical, imaging,
genetic, and fluid-biomarker information in the study of Alzheimer's disease (AD)
progression \citep{mueller2005alzheimer}. ADNI is used as an empirical case study for adjustment recovery in a highly correlated clinical data set. We consider the baseline hippocampal volume from structural MRI as an outcome, a widely used
neurodegeneration biomarker in AD that changes systematically
across the AD continuum \citep{caroli2010dynamics}. Figure~\ref{fig:hippocampal-atrophy} gives the descriptive context in the ADNI baseline cohort: hippocampal volume
declines from cognitively normal participants toward AD. The analysis below
asks how this diagnostic contrast changes after adjustment for baseline
pathology and proxy structure.

\begin{figure}[t!]
    \centering
    \includegraphics[width=0.75\linewidth]{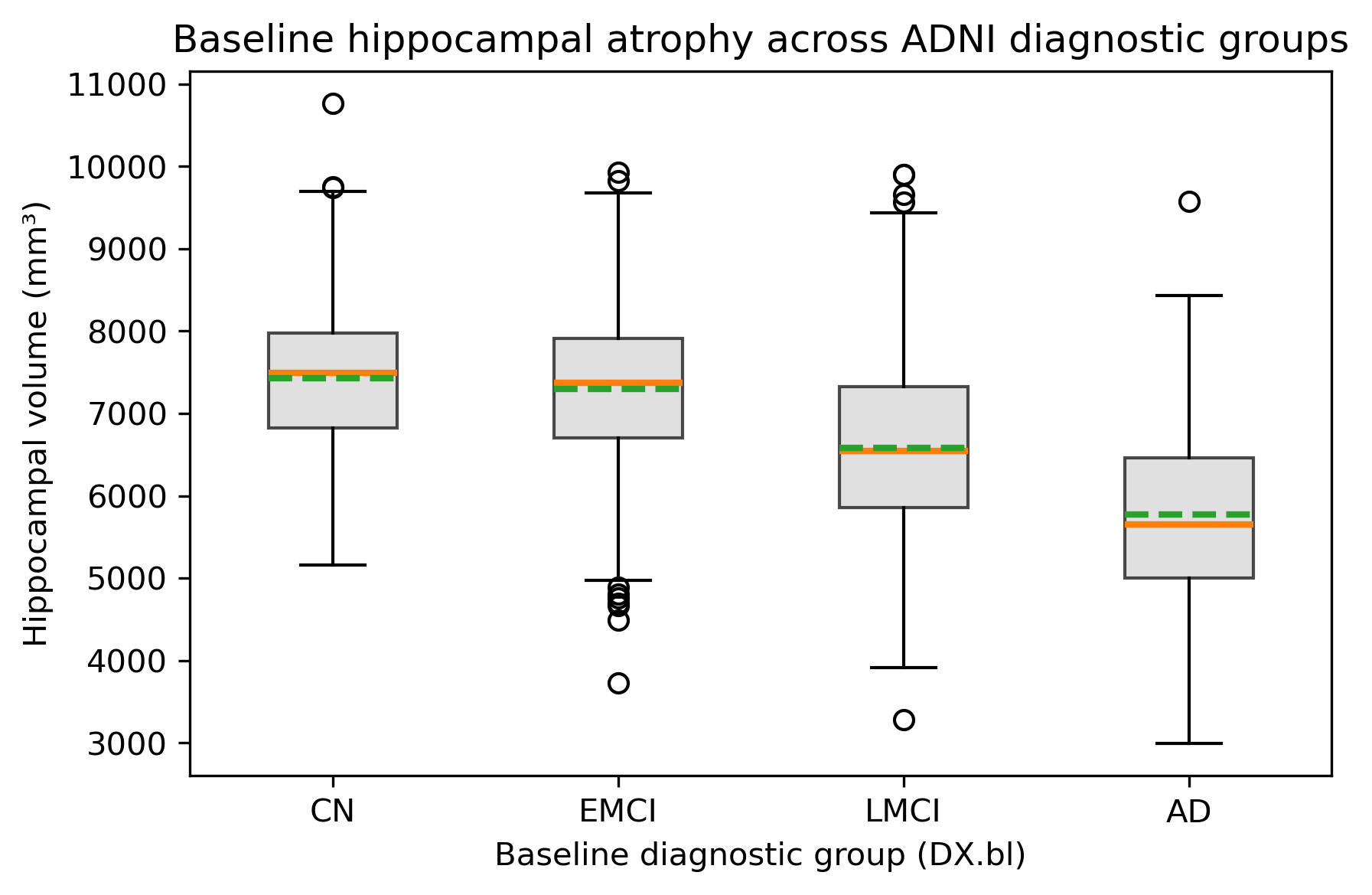}
    \caption{Baseline hippocampal volume across diagnostic groups in ADNI,
    illustrating the expected decline from cognitively normal to AD.}
    \label{fig:hippocampal-atrophy}
\end{figure}

The main empirical question is whether baseline information can recover
the adjustment mimicking the effect of directly using cerebrospinal fluid (CSF) biomarkers. This is a relevant use case because CSF assays are clinically informative but invasive, whereas demographic, genetic, and fluorodeoxyglucose (FDG) PET
measurements are more readily available. We use the full baseline diagnostic continuum with complete data (\(n=802\)):
cognitively normal (CN), early mild cognitive impairment (EMCI), late mild
cognitive impairment (LMCI), and AD. Let \(Y_i\) be hippocampal volume and let
the disease-status coefficients
\(\Delta_{\mathrm{EMCI}}, \Delta_{\mathrm{LMCI}}, \Delta_{\mathrm{AD}}\)
denote linear contrasts against CN. The disease-status indicators are the focal observed exposures in the outcome regression;
hippocampal volume is the outcome; age, education, APOE4, sex, and intracranial volume are commonly measured adjustment variables; and the fitted FA scores are candidate substitute adjustments learned from non-CSF proxies. The direct CSF benchmark additionally adjusts for log amyloid-\(\beta\), total tau,
and phosphorylated tau, with assay-limit strings parsed as their reported
limits before log transformation. This benchmark is not an intervention on CSF
biomarkers; it is an observed pathology adjustment that provides a clinically
interpretable reference for evaluating substitute confounder scores.

The factor models are fit with CSF withheld, using standardized age, education,
APOE4, and FDG as the FA inputs and \(K=2\) factors. Because FDG is the only FA
input not already included among the common outcome covariates, we also report
a measured-proxy comparator that adds FDG directly to the common adjustment.
This guards against attributing to factor analysis what is already achieved by
a single measured proxy. Collapse is diagnosed by the maximum absolute
correlation between each fitted factor score and each FA input; a factor is
flagged when this correlation exceeds \(0.95\). This diagnostic is particularly
important here because several FA inputs are also common outcome covariates: an age-aligned factor would be both a near single-cause summary and a redundant adjustment variable. ADNI's sampling design was designed to resemble an AD clinical-trial cohort, and published comparisons report that ADNI participants are healthier, more highly educated, predominantly white, and more often APOE-\(\varepsilon4\) positive
than community-based cohorts, with limited generalizability to lower-education
and ethnoculturally underrepresented groups
\citep{gianattasio2021generalizability,ashford2022screening}. This induced bias is also studied in our additional APOE4/education enrichment use-case in Appendix~\ref{sec:appendix-adni-diagnostics}. This is why we interpret the estimates reported as internal adjustment-recovery diagnostics within ADNI, not as transported disease effects.

\begin{table}[t!]
\centering
\small
\setlength{\tabcolsep}{6pt}
\caption{Cerebrospinal fluid (CSF)-withheld adjustment recovery in ADNI. The benchmark is the
diagnostic contrast after direct adjustment for log CSF amyloid-\(\beta\), tau,
and \(p\)-tau in addition to age, education, APOE4, sex, and intracranial
volume. The table averages over the EMCI, LMCI, and AD contrasts against CN.
Mean gap reduction is
\(1-|\widehat\Delta-\widehat\Delta_{\mathrm{CSF}}|/
|\widehat\Delta_{\mathrm{common}}-\widehat\Delta_{\mathrm{CSF}}|\), averaged
over the three contrasts.}
\label{tab:adni-csf-recovery}
\begin{tabular}{lccc}
\toprule
Method & Mean abs. gap (\(\mathrm{mm}^3\)) & Mean gap reduction & Collapsed factors \\
\midrule
Measured proxy adjustment & $63$ & $44\%$ & -- \\
Dense FA & $63$ & $43\%$ & -- \\
Sparse FA (NG) & $25$ & $69\%$ & -- \\
Sparse FA (MGP) & $58$ & $47\%$ & -- \\
Sparse FA (CSP) & $49$ & $55\%$ & age \\
\bottomrule
\end{tabular}
\end{table}

\begin{figure}[t!]
    \centering
    \includegraphics[width=0.99\linewidth]{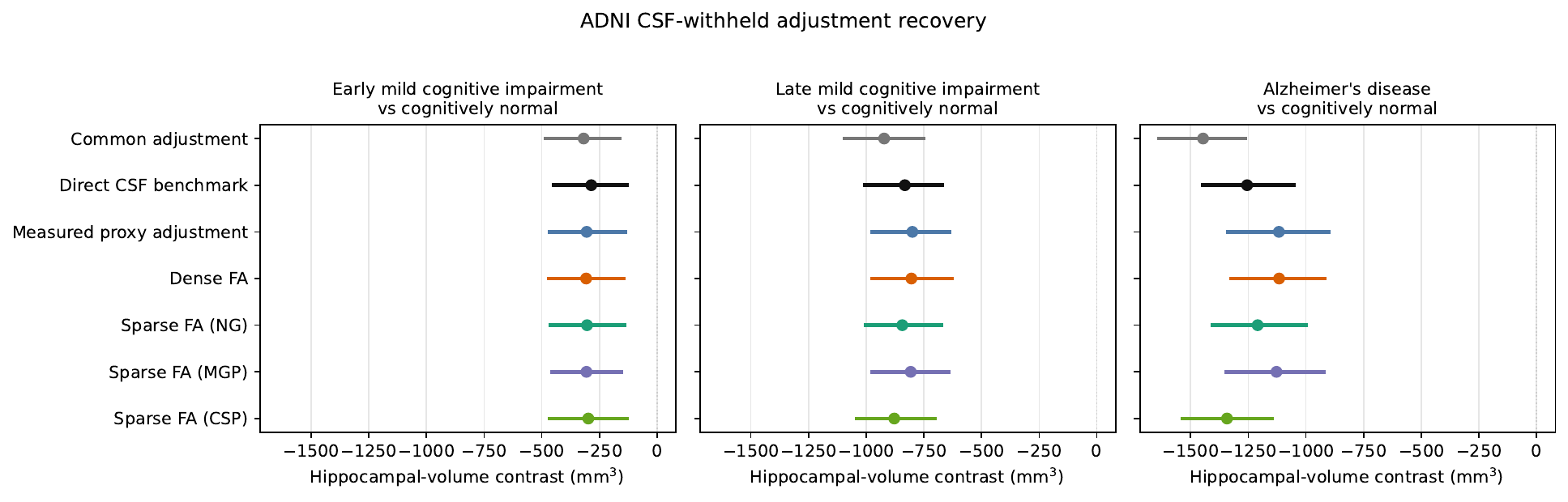}
    \caption{Posterior diagnostic contrasts for baseline hippocampal volume in
    the CSF-withheld ADNI analysis. The direct CSF benchmark adjusts for log
    amyloid-\(\beta\), tau, and \(p\)-tau; substitute confounder rows withhold
    those CSF biomarkers from the factor model.}
    \label{fig:adni-csf-recovery}
\end{figure}

The common adjustment $\hat{\Delta}_{\text{common}}$ and the direct CSF benchmark $\hat{\Delta}_{\text{CSF}}$ differ by approximately
\(33\), \(90\), and \(191\) \(\mathrm{mm}^3\) for the EMCI, LMCI, and AD contrasts,
respectively. Table~\ref{tab:adni-csf-recovery} and
Figure~\ref{fig:adni-csf-recovery} show that direct FDG adjustment and dense FA
give similar recovery of this CSF shift, with average gaps of about
\(63\) \(\mathrm{mm}^3\) from the direct CSF benchmark. The NG score is closest to the CSF benchmark, reducing the average gap to
\(25\) \(\mathrm{mm}^3\) and recovering \(69\%\) of the benchmark shift on average,
without triggering the collapse diagnostic. CSP partially recovers the CSF shift, but one
of its fitted scores has correlation \(0.987\) with age, compared with a
maximum age correlation of \(0.792\) for dense FA. In this benchmark, the NG shrinkage prior gives a substitute score that is closer to direct pathology adjustment than either FDG alone or dense FA, while the CSP row illustrates that shrinkage alone does not prevent a
factor from collapsing onto a single redundant covariate.

Appendix~\ref{sec:appendix-adni-diagnostics} reports an additional APOE4/education enrichment analysis, which is more naturally read as a stress test than as a second positive benchmark. Together with the CSF-withheld analysis, it shows why empirical substitute confounder workflows should report
both benchmark recovery and collapse diagnostics: not every clinically relevant adjustment is recoverable from the remaining measurements.
\section{Summary}

This paper develops a stability-first view of substitute confounders in high-dimensional multi-cause settings. Beyond asking whether a causal functional is identified given an \emph{ideal} latent confounder, we ask when estimation is stable after replacing that confounder by an estimate learned from the causes. Our contribution is to connect (i) outcome-side restrictions that define well-posed causal estimands under shared latent confounding with (ii) structured assignment modeling that yields overlap-preserving, multi-cause latent representations and controlled substitution error. In the Bayesian formulation, posterior concentration on assignment-compatible mechanisms and factor score contraction yield consistency of posterior-averaged regression-adjusted functionals, while the assignment posterior induces the geometry used for projection or regularization.

On the assignment side, Gaussian factor assignment models provide a transparent baseline in which conditional independence and residual variability are explicit. We emphasize two practical failure modes of overly flexible latent representations: \emph{over-encoding}, where $\hat{\vect Z}$ becomes nearly a sufficient statistic for $\vect X$ and conditional variability collapses, and \emph{under-sharing}, where latent coordinates behave like single-cause summaries. We then use structured regularization (via shrinkage priors and multi-cause loading constraints) to concentrate posterior mass on low effective rank and shared dependence patterns. The key quantity is factor score contraction: it directly controls how much error is induced when $\vect Z_i$ is replaced by $\hat{\vect Z}_i$ in downstream regression adjustment.

On the outcome side, we do not try to resolve the general
non-identifiability of arbitrary multi-cause effects from shared latent confounding alone.  Instead, we make the causal target geometry-aware.  In particular, we formalize projected restrictions that target causal variation along directions that are weakly confounded under the learned
geometry. This yields identifiable (or intentionally regularized) causal functionals without requiring that all directions in a high-dimensional treatment space be point-identified. The resulting pipeline is coherent: structured shrinkage stabilizes the learned confounding geometry; outcome restrictions define estimands compatible with that geometry; and plug-in stability bounds translate factor score contraction into controlled causal error.

The practical implication is that substitute confounder workflows should be evaluated not only by fit to $p(\vect X)$, but by whether the learned latent representation preserves conditional variability in $\vect X\mid \hat{\vect Z}$ and yields stable downstream estimation. The geometric perspective also clarifies why ridge-type and direction-specific penalties can be interpreted as overlap preservation rather than purely variance reduction: directions strongly aligned with latent confounding are precisely those where stable estimation is least robust.

Several limitations are intrinsic. First, identification requires outcome-side structure: projected restrictions deliberately exclude or shrink confounded directions, and thus should be interpreted as defining \emph{which causal
questions} are well posed under shared latent confounding. Second, our asymptotic results are ``large-$m$'' and treat factor score contraction as the primitive requirement; finite-sample performance may degrade when $m$ is moderate, loadings are dense, or nuisance structure is comparable in strength to the confounding signal. Third, factor models inherit equivalence and weak identification issues (e.g., rotations and near-redundant factors). Structured
shrinkage reduces these pathologies by breaking rotational symmetry and discouraging near-injective encodings, but sensitivity to prior choice and model misspecification remains and should be assessed via geometry-level diagnostics. Finally, we do not claim that fitting an assignment model to $\vect X$ implies ignorability; any such implication would require assumptions about the full data law. Our focus is on stability and on estimands compatible with learned confounding geometry.

Three extensions are natural: (i) \emph{Tensor and functional analogues:} replace $\vect X_i$ by structured treatments (tensors, trajectories with possible missigness \citep{luo2026wedge}) and couple multilinear or dynamic assignment models with geometry-aware projections; for functional outcomes, combine this with operator-valued regression or Fr\'echet-type targets to define estimands in Sobolev/RKHS settings. (ii) \emph{Kernelized geometry:} replace linear projections by kernel-induced operators, defining confounding geometry in an RKHS and targeting estimands that depend on weakly confounded components in that space; the main technical task is to propagate contraction from latent-score error to operator/projection error. (iii) \emph{Alternative semiparametric outcome restrictions:} beyond projected partially linear structure, one can use index restrictions, additive components, or deliberately regularized estimands defined as solutions to geometry-penalized risk minimization, interpolating between hard projection and full adjustment. Overall, substitute confounders are most defensible when paired with explicit estimands that acknowledge what is and is not identifiable under latent confounding, and with structured regularization that stabilizes learned confounding geometry.

\newpage
\begin{center}
{\LARGE\textbf{Supplementary Materials}}
\end{center}
\appendix
\section{Kallenberg construction of sparse FA models}\label{sec:Kallenberg-construction-proof}

In this section, we prove Lemma \ref{lem:kallenberg-implies-unconf}:
\begin{proof}
Let $U_i=(U_{i1},\dots,U_{im})$ and define the measurable map
\[
  F:\mathcal Z\times [0,1]^m\to \mathcal X, \qquad
  F(\vect z,\vect u)=\bigl(f_1(\vect z,u_1),\dots,f_m(\vect z,u_m)\bigr),
\]
so that $\vect X_i=F(\vect Z_i,U_i)$ almost surely.

Let $\mathcal A\subseteq \mathcal X$ be Borel and let $\mathcal B$ be a measurable
subset of the product space supporting the potential outcome table
$\{Y_i^{(\vect x)}:\vect x\in\mathcal X\}$. We prove conditional independence by showing
\[
  \Prob\!\left(\vect X_i\in \mathcal A,\ \{Y_i^{(\cdot)}\}\in \mathcal B \mid \vect Z_i\right)
  =
  \Prob\!\left(\vect X_i\in \mathcal A\mid \vect Z_i\right)\,
  \Prob\!\left(\{Y_i^{(\cdot)}\}\in \mathcal B \mid \vect Z_i\right),
\]
where $\{Y_i^{(\cdot)}\}$ denotes the entire potential outcome table.

Using $\vect X_i=F(\vect Z_i,U_i)$ and iterated conditional expectation,
\begin{align*}
\Prob\!\left(\vect X_i\in \mathcal A,\ \{Y_i^{(\cdot)}\}\in \mathcal B \mid \vect Z_i\right)
&=
\E\!\left[
\mathbf 1\{F(\vect Z_i,U_i)\in \mathcal A\}\,\mathbf 1\{\{Y_i^{(\cdot)}\}\in\mathcal B\}
\ \middle|\ \vect Z_i\right] \\
&=
\E\!\left[
\E\!\left[\mathbf 1\{F(\vect Z_i,U_i)\in \mathcal A\}\ \middle|\ \vect Z_i,\{Y_i^{(\cdot)}\}\right]
\mathbf 1\{\{Y_i^{(\cdot)}\}\in\mathcal B\}
\ \middle|\ \vect Z_i\right].
\end{align*}

By the conditional exogeneity assumption in Lemma~\ref{lem:kallenberg-implies-unconf},
$U_i\indep \{Y_i^{(\cdot)}\}\mid \vect Z_i$. Hence the inner conditional expectation
does not depend on $\{Y_i^{(\cdot)}\}$:
\[
\E\!\left[\mathbf 1\{F(\vect Z_i,U_i)\in \mathcal A\}\ \middle|\ \vect Z_i,\{Y_i^{(\cdot)}\}\right]
=
\E\!\left[\mathbf 1\{F(\vect Z_i,U_i)\in \mathcal A\}\ \middle|\ \vect Z_i\right]
=
\Prob(\vect X_i\in \mathcal A\mid \vect Z_i).
\]
Substituting back gives
\begin{align*}
\Prob\!\left(\vect X_i\in \mathcal A,\ \{Y_i^{(\cdot)}\}\in \mathcal B \mid \vect Z_i\right)
&=
\Prob(\vect X_i\in \mathcal A\mid \vect Z_i)\,
\E\!\left[\mathbf 1\{\{Y_i^{(\cdot)}\}\in\mathcal B\}\mid \vect Z_i\right] \\
&=
\Prob(\vect X_i\in \mathcal A\mid \vect Z_i)\,
\Prob(\{Y_i^{(\cdot)}\}\in\mathcal B\mid \vect Z_i),
\end{align*}
which is exactly $\vect X_i\indep \{Y_i^{(\vect x)}:\vect x\in\mathcal X\}\mid \vect Z_i$.
The marginal statement $Y_i^{(\vect x)}\indep \vect X_i\mid \vect Z_i$ for each fixed
$\vect x$ follows immediately.
\end{proof}

We also spell out the Gaussian factor representation used in
Section~\ref{sec:gaussian-factor-assignment}:
\begin{proof}
Fix a unit $i$. Under \eqref{eq:sparse-FA} we can write
\begin{align}
  X_{ij} = \lambda_{j\cdot}^\top \vect Z_i + \varepsilon_{ij},
  \qquad j=1,\dots,m,
\end{align}
where $\varepsilon_{ij}\sim \mathcal N(0,\sigma_j^2)$ are mutually independent
across $j$ and independent of $\vect Z_i$. Let $F_{\varepsilon_j}$ be the cdf of
$\mathcal N(0,\sigma_j^2)$ and define
\begin{align}
  U_{ij} := F_{\varepsilon_j}(\varepsilon_{ij}).
\end{align}
By the probability integral transform, $U_{ij}\sim \mathrm U(0,1)$. Mutual
independence across $j$ follows from the mutual independence of
$\{\varepsilon_{ij}\}_{j=1}^m$, and independence of $(U_{i1},\dots,U_{im})$ from
$\vect Z_i$ follows from independence of $\boldsymbol\varepsilon_i$ and
$\vect Z_i$. Moreover,
\begin{align}
  X_{ij}
  = \lambda_{j\cdot}^\top \vect Z_i + F_{\varepsilon_j}^{-1}(U_{ij})
  =: f_j(\vect Z_i,U_{ij}),
\end{align}
which is a Kallenberg construction of $(X_{i1},\dots,X_{im})$.

If in addition \eqref{eq:exogeneity-U} holds, then conditional on $\vect Z_i$ the
vector $(X_{i1},\dots,X_{im})$ is a measurable function of $(U_{i1},\dots,U_{im})$
and $\vect Z_i$, with $(U_{i1},\dots,U_{im})$ independent of the entire potential outcome
table given $\vect Z_i$. Therefore
\begin{align}
  (X_{i1},\dots,X_{im})
  \;\indep\;
  \{Y_i^{(\vect x)}:\vect x\in\mathcal X\}
  \;\big|\; \vect Z_i,
\end{align}
establishing weak unconfoundedness.

Finally, Assumption~\ref{assump:SC} is not needed for the construction, but it ensures that each latent coordinate corresponds to variation shared by at least two causes (no single-cause factors), matching the multi-cause confounding interpretation.
\end{proof}

\section{Stability under a consistent substitute}\label{sec:stability-under-consistent-substitute}
This appendix gives a conditional version of the simple stability idea used in
Proposition~\ref{prop:plugin-stability}.  The statement is deliberately restricted
to treatment strata with positive probability, such as discrete treatments or
coarsened treatment cells.  For continuous treatments, the same logic should be
applied to shrinking neighbourhoods of \(\vect x\), with the corresponding
regular conditional versions stated explicitly.

\begin{proposition}[Conditional stability on positive-probability strata]
\label{prop:ignorability-preserved}
Fix a treatment stratum \(A_{\vect x}:=\{\vect X_i=\vect x\}\) with
\(\Prob(A_{\vect x})>0\).  Let \(h\) be bounded.  Suppose:
\begin{enumerate}
\item \emph{Weak unconfoundedness}: for every treatment vector \(\vect x\),
\[
Y_i^{(\vect x)} \indep \vect X_i \mid \vect Z_i .
\]
\item \emph{Conditional score accuracy within the stratum}:
\[
\E\!\left[\|\hat{\vect Z}_i-\vect Z_i\|\mid A_{\vect x}\right]\to0 .
\]
\item \emph{Smooth conditional mean}: the function
\[
\phi_{\vect x}(\vect z)
:=\E\!\left[h(Y_i^{(\vect x)})\mid \vect Z_i=\vect z\right]
\]
is Lipschitz: for some finite \(L_h(\vect x)\),
\[
|\phi_{\vect x}(\vect z)-\phi_{\vect x}(\vect z')|
\le L_h(\vect x)\|\vect z-\vect z'\| .
\]
\end{enumerate}
Then
\begin{align}
&\E\!\left[
  \left|
  \E\!\left[h(Y_i^{(\vect x)})\mid A_{\vect x},\hat{\vect Z}_i\right]
  -
  \E\!\left[h(Y_i^{(\vect x)})\mid A_{\vect x},\vect Z_i\right]
  \right|
  \Bigm| A_{\vect x}
\right]
\notag\\
&\hspace{4em}\le
2L_h(\vect x)\,
\E\!\left[\|\hat{\vect Z}_i-\vect Z_i\|\mid A_{\vect x}\right]
\longrightarrow 0 .
\label{eq:conditional-stability}
\end{align}
\end{proposition}

\begin{proof}
Fix $\vect x$ and a bounded measurable function $h$.
By weak unconfoundedness,
\[
\E\!\left[h(Y_i^{(\vect x)})\mid A_{\vect x},\vect Z_i\right]
=\E\!\left[h(Y_i^{(\vect x)})\mid \vect Z_i\right]
=\phi_{\vect x}(\vect Z_i).
\]
Let \(\mathcal G=\sigma(A_{\vect x},\hat{\vect Z}_i)\).  Since
\(\phi_{\vect x}(\hat{\vect Z}_i)\) is \(\mathcal G\)-measurable,
\begin{align*}
&\left|
\E[\phi_{\vect x}(\vect Z_i)\mid\mathcal G]
-\phi_{\vect x}(\vect Z_i)
\right| \\
&\qquad\le
\left|\E[\phi_{\vect x}(\vect Z_i)-\phi_{\vect x}(\hat{\vect Z}_i)\mid\mathcal G]\right|
+
\left|\phi_{\vect x}(\hat{\vect Z}_i)-\phi_{\vect x}(\vect Z_i)\right| .
\end{align*}
Taking conditional expectations given \(A_{\vect x}\) and using the Lipschitz
condition gives the displayed bound.  The bound converges to zero by conditional
score accuracy.
\end{proof}

\section{Consistency of regression adjusted mean potential outcomes}\label{sec:proof-mu-consistency}
Next, we formally prove Theorem~\ref{thm:mu-consistency-ATE}.
\begin{proof}
Fix $\vect x\in\mathcal X$. Let $P_0$ denote the data-generating law of
$\big(\{Y^{(\vect x)}:\vect x\in\mathcal X\},\vect X,\vect Z\big)$ for a generic unit,
and write $\E_0,\Prob_0$ for expectation and probability under $P_0$.  In this
proof we write $\hat{\vect Z}_i$ for $\hat{\vect Z}_{i,nm}$ to simplify notation.

By iterated expectations, latent ignorability, positivity, and consistency,
\begin{equation}
  \mu(\vect x)
  =
  \E_0\bigl[ r(\vect x,\vect Z_i) \bigr].
  \label{eq:mu-x-id}
\end{equation}
Indeed,
\[
  \E_0\{Y_i^{(\vect x)}\mid \vect Z_i=\vect z\}
  =
  \E_0\{Y_i\mid \vect X_i=\vect x,\vect Z_i=\vect z\}
  =r(\vect x,\vect z).
\]

From \eqref{eq:mu-hat-ATE},
\begin{equation}
  \hat\mu_n(\vect x)-\mu(\vect x)
  =
  \underbrace{\frac1n\sum_{i=1}^n
    \{\hat r_n(\vect x,\hat{\vect Z}_i)-r(\vect x,\hat{\vect Z}_i)\}}
    _{\text{(I)}}
  +
  \underbrace{\left\{\frac1n\sum_{i=1}^n r(\vect x,\hat{\vect Z}_i)
      -\E_0[r(\vect x,\vect Z_i)]\right\}}
    _{\text{(II)}} .
  \label{eq:mu-decomp}
\end{equation}

For (I), the average calibration condition in Theorem~\ref{thm:mu-consistency-ATE}
gives
\[
  \E_0|\text{(I)}|
  \le
  \E_0\left[\frac1n\sum_{i=1}^n
  |\hat r_n(\vect x,\hat{\vect Z}_i)-r(\vect x,\hat{\vect Z}_i)|\right]
  =o(1).
\]
Thus \(\text{(I)}\to0\) in \(\Prob_0\) by Markov's inequality.

For (II), write
\begin{equation}
  \text{(II)}=
  \underbrace{\frac1n\sum_{i=1}^n
  \{r(\vect x,\hat{\vect Z}_i)-r(\vect x,\vect Z_i)\}}
  _{\text{(IIa)}}
  +
  \underbrace{\left\{\frac1n\sum_{i=1}^n r(\vect x,\vect Z_i)
  -\E_0[r(\vect x,\vect Z_i)]\right\}}
  _{\text{(IIb)}} .
  \label{eq:II-split}
\end{equation}
Condition (C4) gives \(\text{(IIb)}\to0\) in \(L^1(P_0)\), and hence in
\(\Prob_0\), along the joint sequence.  For (IIa), Lipschitz continuity and Cauchy-Schwarz imply
\begin{align}
  \E_0|\text{(IIa)}|
  &\le
  L(\vect x)\,
  \E_0\left[\frac1n\sum_{i=1}^n\|\hat{\vect Z}_i-\vect Z_i\|\right]
  \nonumber\\
  &\le
  L(\vect x)\,
  \E_0\left[\left\{\frac1n\sum_{i=1}^n
  \|\hat{\vect Z}_i-\vect Z_i\|^2\right\}^{1/2}\right]
  \nonumber\\
  &\le
  L(\vect x)\,
  \left(\E_0\left[\frac1n\sum_{i=1}^n
  \|\hat{\vect Z}_i-\vect Z_i\|^2\right]\right)^{1/2}
  \le L(\vect x)\varepsilon_{n,m}.
  \label{eq:IIa-mean-bound}
\end{align}
Therefore \(\text{(IIa)}\to0\) in \(\Prob_0\) whenever
\(\varepsilon_{n,m}\to0\).  Combining the three parts proves
\[
  \hat\mu_n(\vect x)\xrightarrow{\Prob_0}\mu(\vect x)
\]
along the joint sequence.

Taking expectations in \eqref{eq:mu-decomp} and using the triangle inequality,
\begin{align}
  \bigl|\E_0[\hat\mu_n(\vect x)]-\mu(\vect x)\bigr|
  &\le
  \E_0|\text{(I)}|+\E_0|\text{(IIa)}|+\E_0|\text{(IIb)}|.
\end{align}
The first term is \(o(1)\) by average calibration, and the second term is bounded
by \(L(\vect x)\varepsilon_{n,m}\) by \eqref{eq:IIa-mean-bound}.  The third term is \(o(1)\) by condition (C4).  This yields
\[
  \bigl|\E_0[\hat\mu_n(\vect x)]-\mu(\vect x)\bigr|
  \le L(\vect x)\varepsilon_{n,m}+o(1),
\]
where the remainder is along the same joint sequence.
\end{proof}

\section{Outcome-side restrictions and identifiability: proofs and additional examples}
\label{sec:appendix-outcome-identifiability}

This appendix provides proofs for the identification results in
Section~\ref{subsec:outcome-identifiability} and records additional outcome-side restrictions that yield identifiable or well-posed (regularized) causal estimands under shared latent confounding.

Throughout, consider the Gaussian factor assignment model
\begin{equation}
\label{eq:appendix-factor}
  \vect X_i = \Lambda Z_i + \varepsilon_i,
  \qquad
  Z_i \sim \mathcal N_H(0,I_H),
  \qquad
  \varepsilon_i \sim \mathcal N_m(0,\Psi),
\end{equation}
where $\Psi \succ 0$ is diagonal and $\varepsilon_i \indep Z_i$.
Let $\mathcal C := \operatorname{span}(\Lambda) \subset \mathbb R^m$, and
denote by $P_{\mathcal C}$ and $P_{\mathcal C}^\perp := I_m - P_{\mathcal C}$
the orthogonal projections onto $\mathcal C$ and its orthogonal complement. Under the factor assignment model \eqref{eq:appendix-factor} with
$\mathcal C=\operatorname{span}(\Lambda)$, we have $P_{\mathcal C}^\perp\Lambda=0$
and therefore
\begin{equation}
\label{eq:orth-identity}
P_{\mathcal C}^\perp \vect X_i
=
P_{\mathcal C}^\perp(\Lambda Z_i+\varepsilon_i)
=
P_{\mathcal C}^\perp \varepsilon_i .
\end{equation}
Since $\varepsilon_i \indep Z_i$ and $\mathbb E[\varepsilon_i]=0$, it follows that
\begin{equation}
\label{eq:orth-mean}
\mathbb E[P_{\mathcal C}^\perp \vect X_i \mid Z_i]=0,
\qquad
\Var(P_{\mathcal C}^\perp \vect X_i)=P_{\mathcal C}^\perp\Psi P_{\mathcal C}^\perp.
\end{equation}
Under Gaussianity, \eqref{eq:orth-identity} also implies
$P_{\mathcal C}^\perp \vect X_i \indep Z_i$.

%------------------------------------------------------------
\subsection{Proof of Proposition~\ref{prop:proj-identification}}
\label{sec:appendix-proof-proj-id}

\begin{proof}
Let $B_{\mathcal C}\in\mathbb R^{m\times r}$ be an orthonormal basis for
$\mathcal C^\perp$, let $W:=B_{\mathcal C}^\top\vect X$, and let
$U:=P_{\mathcal C}^\perp\vect X=B_{\mathcal C}W$. By assumption,
$\E[U\mid\vect Z]=0$, hence $\E[W\mid\vect Z]=0$. Write
$\beta_\perp=P_{\mathcal C}^\perp\beta$ and
$\gamma:=B_{\mathcal C}^\top\beta_\perp$, so that
$\beta_\perp^\top U=\gamma^\top W$.

By Assumption~\ref{ass:proj-outcome},
\begin{align}
  \E[Y\mid\vect X,\vect Z]=g(\vect Z)+\gamma^\top W .
\end{align}
Taking covariance with $W$ and using iterated expectation gives
\begin{align*}
  \Cov(W,Y)
  &=\Cov\!\left(W,\E[Y\mid\vect X,\vect Z]\right) \\
  &=\Cov\!\left(W,g(\vect Z)+\gamma^\top W\right) \\
  &=\Cov(W,g(\vect Z))+\Var(W)\gamma.
\end{align*}
The first term is zero because
$\E[W\mid\vect Z]=0$:
\[
  \E\{Wg(\vect Z)\}=\E\{\E[W\mid\vect Z]g(\vect Z)\}=0 .
\]
Since $\Var(W)$ is nonsingular, \eqref{eq:gamma-proj-id} follows:
\[
  \gamma=\Var(W)^{-1}\Cov(W,Y),
  \qquad
  \beta_\perp=B_{\mathcal C}\gamma.
\]
Equivalently, because $\Var(U)=B_{\mathcal C}\Var(W)B_{\mathcal C}^\top$ has
range $\mathcal C^\perp$,
\[
  \Var(U)^+=B_{\mathcal C}\Var(W)^{-1}B_{\mathcal C}^\top,
  \qquad
  \beta_\perp=\Var(U)^+\Cov(U,Y).
\]
All quantities on the right are functions of the observed law of $(Y,\vect X)$
and the fixed projection $P_{\mathcal C}^\perp$.

For the projected mean, Assumption~\ref{ass:proj-outcome} gives
\begin{align}
  \E[Y^{(\vect x)}\mid\vect Z]
  =g(\vect Z)+\beta_\perp^\top P_{\mathcal C}^\perp\vect x .
\end{align}
Taking expectations over $\vect Z$ gives
\begin{align}
  \E[Y^{(\vect x)}]
  =\E[g(\vect Z)]+\beta_\perp^\top P_{\mathcal C}^\perp\vect x .
\end{align}
Finally,
\begin{align}
  \E[Y-\beta_\perp^\top U]
  =\E[g(\vect Z)],
\end{align}
since $\E[U\mid\vect Z]=0$. Substitution yields
\eqref{eq:mu-perp-id}. The expression is denoted $\mu_\perp(\vect x)$ in the
main text to emphasize that only variation in $\mathcal C^\perp$ is identified;
under Assumption~\ref{ass:proj-outcome} it coincides with the corresponding mean
potential outcome.
\end{proof}

\subsection{Soft geometry regularization of outcome effects}
\label{sec:appendix-soft-geo}

In contrast to the projected restrictions of
Section~\ref{subsec:outcome-identifiability}, which exclude confounded
directions from the causal estimand, one may instead allow effects along
the confounding subspace while explicitly shrinking them. This yields
well-posed but intentionally regularized causal estimands.

Suppose the outcome regression is linear in the causes conditional on the
latent confounder:
\begin{equation}
\label{eq:soft-linear}
\mathbb E[Y \mid \vect X, \vect Z]
=
g(\vect Z) + \beta^\top \vect X,
\end{equation}
where $g$ is unrestricted. Under shared latent confounding, $\beta$ is
not point-identified without additional structure. We therefore define a
geometry-aware estimand by penalizing directions aligned with the
confounding subspace.

Let $\mathcal C \subset \mathbb R^m$ denote the confounding subspace and
$P_{\mathcal C}$ the orthogonal projection onto $\mathcal C$. For
$\lambda>0$, define
\begin{equation}
\label{eq:beta-lambda-def}
(\alpha_\lambda,\beta_\lambda)
:=
\arg\min_{\alpha\in\mathbb R,\,\beta \in \mathbb R^m}
\left\{
\mathbb E\!\left[(Y-\alpha-\beta^\top \vect X)^2\right]
+
\lambda\,\beta^\top P_{\mathcal C}\beta
\right\}.
\end{equation}
The intercept is unpenalized.  The criterion shrinks coefficients aligned with
$\mathcal C$ while leaving directions in $\mathcal C^\perp$ unpenalized.

\begin{proposition}[Well-posed geometry-regularized estimand]
\label{prop:soft-geo}
Suppose the assignment model satisfies
$\vect X = \Lambda \vect Z + \varepsilon$ with $\Var(\varepsilon)=\Psi \succ 0$.
Then for every $\lambda>0$:
\begin{enumerate}[leftmargin=1.2em]
\item the slope minimizer $\beta_\lambda$ in \eqref{eq:beta-lambda-def} exists and
      is unique, with intercept $\alpha_\lambda=\E(Y)-\beta_\lambda^\top\E(\vect X)$;
\item $\beta_\lambda$ satisfies the normal equation
      \[
      \bigl(\Var(\vect X) + \lambda P_{\mathcal C}\bigr)\beta_\lambda
      =
      \Cov(\vect X,Y),
      \]
      where all variances, covariances, and expectations are taken under the
      observed law of $(\vect X,Y)$;
\item the penalty term has the Bayesian interpretation of a Gaussian prior on
      the confounded component $P_{\mathcal C}\beta$ with precision proportional
      to $\lambda$, together with a flat (improper) prior on
      $P_{\mathcal C}^\perp\beta$.  If a proper Gaussian prior is desired in all
      directions, one may add an isotropic ridge term $\lambda_0\|\beta\|^2$ to
      \eqref{eq:beta-lambda-def}; the resulting prior precision is proportional
      to $\lambda_0 I_m+\lambda P_{\mathcal C}$.
\end{enumerate}
\end{proposition}

\begin{proof}
For fixed $\beta$, the unpenalized intercept is
$\alpha_\beta=\E(Y)-\beta^\top\E(\vect X)$.  Profiling out the intercept reduces
\eqref{eq:beta-lambda-def} to minimizing
\[
\Var(Y-\beta^\top\vect X)+\lambda\beta^\top P_{\mathcal C}\beta .
\]
The Hessian in $\beta$ is proportional to
\[
\Var(\vect X) + \lambda P_{\mathcal C}
=
\Lambda\Lambda^\top + \Psi + \lambda P_{\mathcal C},
\]
which is positive definite because $\Psi\succ0$. Hence the slope minimizer is
unique.  Differentiating the profiled objective with respect to $\beta$ and
setting the gradient to zero yields the stated normal equation.  The Bayesian
interpretation follows from standard conjugacy in Gaussian linear regression:
a quadratic penalty only on $P_{\mathcal C}\beta$ corresponds to a Gaussian prior
on that projected component and a flat prior on the orthogonal component.  Adding
an isotropic ridge term makes the prior proper in every direction and gives the
stated precision form.
\end{proof}

The coefficient $\beta_\lambda$ is not a point-identified causal effect
vector. Instead, it defines a geometry-regularized estimand that shrinks
directions of the treatment space aligned with latent confounding. As
$\lambda \to \infty$, effects along $\mathcal C$ are suppressed, recovering
the hard projection considered in
Section~\ref{subsec:outcome-identifiability}; for finite $\lambda$, effects
along $\mathcal C$ are allowed but explicitly regularized.

In the main text, the confounding geometry $\mathcal C$ is either assumed known
or identified under additional assignment-model conditions. In the Bayesian
setting it may be represented for regularization by the posterior operator
$H_n \approx \mathbb E[\Lambda\Lambda^\top \mid \vect X_{1:n}]$. Replacing
$P_{\mathcal C}$ by $H_n$ yields a direction-specific ridge penalty aligned
with posterior confounding geometry, but this substitution should be interpreted
as regularization unless subspace-identification assumptions are imposed.
This mirrors the duality between imbalance penalties and ridge adjustment
formalized in augmented balancing weights \cite{bruns2025augmented}.

\section{Proof of Theorem~\ref{thm:bayes-mu-consistency}}
\label{sec:appendix-bayes-consistency}

This appendix proves Theorem~\ref{thm:bayes-mu-consistency} under the explicit
posterior concentration, calibration, Lipschitz, and envelope conditions stated
in the main text. We reuse the plug-in stability principle of
Proposition~\ref{prop:plugin-stability} to control the effect of replacing
$\vect Z_i$ by posterior draw-specific factor scores
$\hat{\vect Z}_i(\theta_X)$.

To target $\mu(\vect x)=\E[Y_i^{(\vect x)}]$ in the latent confounder framework,
the relevant regression functional is $\E[r(\vect x,\vect Z_i)]$ rather than a conditional functional that integrates against $p(\vect z\mid\vect X=\vect x)$. Accordingly, throughout this appendix, the Bayesian procedure produces the posterior-averaged regression-adjusted functional
\begin{equation}
\label{eq:bayes-target-appendix}
  \mu_n^{\mathrm{pp}}(\vect x)
  :=
  \int
    \Biggl\{
      \frac{1}{n}\sum_{i=1}^n
      r_{\theta_Y}\!\bigl(\vect x,\hat{\vect Z}_i(\theta_X)\bigr)
    \Biggr\}\,
  \widetilde\Pi_n(d\theta),
\end{equation}
where $\widetilde\Pi_n$ denotes either the joint posterior or the cut posterior used by the
procedure, and $\hat{\vect Z}_i(\theta_X)$ is a measurable factor score summary under the assignment draw $\theta_X$, for example $\E_{\theta_X}(\vect Z_i\mid\vect X_i)$. The fixed posterior mean score is a special case obtained by replacing $\hat{\vect Z}_i(\theta_X)$ with its assignment posterior average.

Let $\theta_0$ denote the true data-generating parameter; throughout this
appendix $P_{\theta_0}$ and $\E_{\theta_0}$ are the same data-generating
probability and expectation denoted by $P_0$ and $\E_0$ in the main text. Define
\begin{align}
  r_{\theta_0}(\vect x,\vect z)
  := \E_{\theta_0}\!\left[Y_i \mid \vect X_i=\vect x,\vect Z_i=\vect z\right].
\end{align}
Under latent ignorability and positivity (Assumption~\ref{ass:latent-ign-positivity}),
the mean potential outcome satisfies the standard representation
\begin{equation}
\label{eq:mu-id-appendix}
  \mu(\vect x)
  =
  \E_{\theta_0}\!\left[r_{\theta_0}(\vect x,\vect Z_i)\right].
\end{equation}
Theorem~\ref{thm:bayes-mu-consistency} assumes the corresponding oracle
empirical-average condition for
\(r_{\theta_0}(\vect x,\vect Z_i)\); the stronger sufficient condition
\(\sup_m\E_{\theta_0}\{r_{\theta_0}(\vect x,\vect Z_i)^2\}<\infty\) is often the
simplest way to verify it along a joint sequence.

Theorem~\ref{thm:bayes-mu-consistency} is stated at a high level.  For a
transparent proof, we make explicit the posterior concentration needed
to ensure that posterior-averaged regression functionals are close to the
truth.

\begin{assumption}[Posterior concentration on assignment-compatible mechanisms]
\label{ass:bayes-conc-ID}
Let $\mathcal M_{\mathrm{A}}$ be as in Definition~\ref{def:posterior-ID}. Then
\begin{align}
  \widetilde\Pi_{n,X}\!\bigl(\mathcal M_{\mathrm{A}}^c\bigr)
  \xrightarrow{P_{\theta_0}} 0 .
\end{align}
\end{assumption}

\noindent
Assumption~\ref{ass:bayes-conc-ID} is not used to control the regression
error directly; rather, it justifies that the posterior mass asymptotically
concentrates on assignment mechanisms for which (i) the substitute confounderrepresentation is well-defined and (ii) the factor score summaries $\hat{\vect Z}_i$ correspond to \emph{assignment-compatible} latent variables.  In particular, it rules out posterior mass on degenerate or near-invertible latent representations that would undermine latent overlap
and invalidate substitute confounder adjustment. The convergence argument below is then driven by the contraction of factor scores and calibration of the outcome regression, conditional on the posterior concentrating on such well-behaved mechanisms.

\begin{assumption}[Outcome posterior calibration on relevant latent values]
\label{ass:bayes-contraction-theta0}
For each fixed $\vect x\in\mathcal X$,
\begin{equation}
\label{eq:bayes-m-unif}
  \E_{\theta_0}\!\left[
  \int \frac1n\sum_{i=1}^n
  \bigl|r_{\theta_Y}(\vect x,\vect Z_i)-r_{\theta_0}(\vect x,\vect Z_i)\bigr|\,
  \widetilde\Pi_n(d\theta)
  \right]
  \longrightarrow 0 .
\end{equation}
\end{assumption}

\noindent
Assumption~\ref{ass:bayes-contraction-theta0} encodes the required posterior calibration of the outcome regression on the latent values that actually occur in the sample.  This is weaker than asking for uniform accuracy over the whole latent space $\mathcal Z$, which may be unrealistic when $\mathcal Z=\mathbb R^H$.

\begin{assumption}[Posterior factor score contraction]
\label{ass:bayes-score-contract}
There exists a sequence $\varepsilon_{n,m}\downarrow0$ along the joint asymptotic sequence such that the posterior
factor score summaries satisfy
\[
  \E_{\theta_0}\!\left[
  \int \frac1n\sum_{i=1}^n
  \|\hat{\vect Z}_i(\theta_X)-\vect Z_i\|^2\,
  \widetilde\Pi_n(d\theta)
  \right]
  \le \varepsilon_{n,m}^2+o(1),
  \qquad
  \varepsilon_{n,m}\to0 \quad ((n,m)\to\infty).
\]
If the factor model is identifiable only up to rotation, the norm is interpreted
after imposing the chosen orientation constraint, or with an infimum over the
admissible orthogonal transformations.
\end{assumption}

\begin{assumption}[Lipschitz outcome regression on posterior-typical parameter sets]
\label{ass:bayes-lip}
Fix $\vect x\in\mathcal X$.  There exists $L(\vect x)<\infty$ and sets
$\Theta_n(\vect x)\subseteq \Theta$ such that
$\widetilde\Pi_n(\Theta_n(\vect x)^c)\xrightarrow{P_{\theta_0}}0$ and, for all
$\theta\in\Theta_n(\vect x)$ and all $\vect z,\vect z'\in\mathcal Z$,
\begin{align}
  \bigl|r_{\theta_Y}(\vect x,\vect z)-r_{\theta_Y}(\vect x,\vect z')\bigr|
  \le L(\vect x)\,\|\vect z-\vect z'\|.
\end{align}
\end{assumption}

\begin{assumption}[Square-integrable envelope]
\label{ass:bayes-envelope}
Fix $\vect x\in\mathcal X$. There exist nonnegative random variables
$M_{i,n}(\vect x)$, allowed to depend on the data-generating triangular array and
on the fitted-score process but not on the posterior draw once the data are
fixed, such that, with $P_{\theta_0}$-probability one, the following bound holds
for all posterior parameter values $\theta$ in the support of $\widetilde\Pi_n$:
\begin{align}
  \sup_{\theta\in\operatorname{supp}(\widetilde\Pi_n)}
  \left\{
  \bigl|r_{\theta_Y}(\vect x,\vect Z_i)\bigr|
  +
  \bigl|r_{\theta_Y}(\vect x,\hat{\vect Z}_i(\theta_X))\bigr|
  \right\}
  \le M_{i,n}(\vect x),
  \qquad
  \sup_{i,n} \E_{\theta_0}\!\bigl[M_{i,n}(\vect x)^2\bigr]<\infty.
\end{align}
In addition, \(n^{-1}\sum_{i=1}^n M_{i,n}(\vect x)\) is uniformly integrable.
\end{assumption}

\noindent
Assumption~\ref{ass:bayes-envelope} is a standard domination condition stated
globally over posterior draws, so that the negligible tail
\(\Theta_n(\vect x)^c\) can be controlled. The Lipschitz condition in
Assumption~\ref{ass:bayes-lip} is required only on the posterior-typical set
\(\Theta_n(\vect x)\). If the envelope were imposed only on
\(\Theta_n(\vect x)\), an additional tail-envelope condition would be needed.

\subsection{Proof of Theorem~\ref{thm:bayes-mu-consistency}}

\begin{proof}[Proof of Theorem~\ref{thm:bayes-mu-consistency}]
Fix $\vect x\in\mathcal X$.
Define the auxiliary functional that uses the \emph{true} latent
confounder in the regression adjustment,
\begin{equation}
\label{eq:bayes-oracle}
  \mu_n^{\mathrm{or}}(\vect x)
  :=
  \int
    \Biggl\{
      \frac{1}{n}\sum_{i=1}^n r_{\theta_Y}\!\bigl(\vect x,\vect Z_i\bigr)
    \Biggr\}\,
  \widetilde\Pi_n(d\theta).
\end{equation}
We decompose
\begin{equation}
\label{eq:bayes-decomp}
  \mu_n^{\mathrm{pp}}(\vect x)-\mu(\vect x)
  =
  \underbrace{
    \bigl\{\mu_n^{\mathrm{pp}}(\vect x)-\mu_n^{\mathrm{or}}(\vect x)\bigr\}
  }_{(A)}
  +
  \underbrace{
    \bigl\{\mu_n^{\mathrm{or}}(\vect x)-\mu(\vect x)\bigr\}
  }_{(B)}.
\end{equation}

\paragraph{Step 1:} 
Write
\begin{align}
  (A)
  =
  \int
    \Biggl\{
      \frac{1}{n}\sum_{i=1}^n
      \Bigl(
        r_{\theta_Y}(\vect x,\hat{\vect Z}_i(\theta_X))-r_{\theta_Y}(\vect x,\vect Z_i)
      \Bigr)
    \Biggr\}\,
  \widetilde\Pi_n(d\theta).
\end{align}
Split the posterior integral over $\Theta_n(\vect x)$ and
$\Theta_n(\vect x)^c$.  On $\Theta_n(\vect x)$, Assumption~\ref{ass:bayes-lip}
gives, for each $i$,
\begin{align}
  \bigl|r_{\theta_Y}(\vect x,\hat{\vect Z}_i(\theta_X))-r_{\theta_Y}(\vect x,\vect Z_i)\bigr|
  \le L(\vect x)\,\|\hat{\vect Z}_i(\theta_X)-\vect Z_i\|.
\end{align}
Hence
\begin{align}
\label{eq:A-split}
  \bigl|(A)\bigr|
  &\le
  \int_{\Theta_n(\vect x)}
  \Biggl\{
    \frac{1}{n}\sum_{i=1}^n
    L(\vect x)\,\|\hat{\vect Z}_i(\theta_X)-\vect Z_i\|
  \Biggr\}\,\widetilde\Pi_n(d\theta)
  \nonumber\\
  &\quad+\;
  \int_{\Theta_n(\vect x)^c}
  \Biggl\{
    \frac{1}{n}\sum_{i=1}^n
    2M_{i,n}(\vect x)
  \Biggr\}\,\widetilde\Pi_n(d\theta).
\end{align}
By Assumption~\ref{ass:bayes-envelope},
\begin{align}
  \int_{\Theta_n(\vect x)^c}
    \Biggl\{
      \frac{1}{n}\sum_{i=1}^n
      2M_{i,n}(\vect x)
    \Biggr\}\,\widetilde\Pi_n(d\theta)
  \le 2\left\{\frac1n\sum_{i=1}^n M_{i,n}(\vect x)\right\}\,\widetilde\Pi_n\!\bigl(\Theta_n(\vect x)^c\bigr).
\end{align}
Let
\(A_n(\vect x):=n^{-1}\sum_{i=1}^n M_{i,n}(\vect x)\) and
\(B_n(\vect x):=\widetilde\Pi_n(\Theta_n(\vect x)^c)\).  Assumption~\ref{ass:bayes-envelope}
implies \(\sup_n\E_{\theta_0}\{A_n(\vect x)^2\}<\infty\), because
\((n^{-1}\sum_i M_{i,n})^2\le n^{-1}\sum_i M_{i,n}^2\).  Hence \(A_n(\vect x)\) is
uniformly integrable.  Since \(0\le B_n(\vect x)\le1\), the products
\(A_n(\vect x)B_n(\vect x)\) are also uniformly integrable.  The same
second-moment bound gives \(A_n(\vect x)=O_{P_{\theta_0}}(1)\), while
\(B_n(\vect x)\to0\) in probability; hence
\(A_n(\vect x)B_n(\vect x)\to0\) in probability.  Vitali's theorem then gives
\begin{align}
\label{eq:A-tail-L2}
  \E_{\theta_0}\!\left[2A_n(\vect x)B_n(\vect x)\right]=o(1).
\end{align}

For the first term in \eqref{eq:A-split}, Jensen and Cauchy--Schwarz yield
\begin{align}
&\E_{\theta_0}\!\left[
  \int \frac{1}{n}\sum_{i=1}^n
  \|\hat{\vect Z}_i(\theta_X)-\vect Z_i\|\,\widetilde\Pi_n(d\theta)
  \right]
  \\
&\qquad\le
  \left(
  \E_{\theta_0}\!\left[
  \int \frac{1}{n}\sum_{i=1}^n
  \|\hat{\vect Z}_i(\theta_X)-\vect Z_i\|^2\,\widetilde\Pi_n(d\theta)
  \right]
  \right)^{1/2}
  \le \varepsilon_{n,m}+o(1),
\end{align}
by Assumption~\ref{ass:bayes-score-contract}. Combining with
\eqref{eq:A-tail-L2} gives
\begin{equation}
\label{eq:A-bound}
  \E_{\theta_0}\bigl[\,|(A)|\,\bigr]
  \le
  L(\vect x)\,\varepsilon_{n,m} + o(1).
\end{equation}
In particular, Markov's inequality implies $(A)=o_{P_{\theta_0}}(1)$ when
$\varepsilon_{n,m}\to0$.

\paragraph{Step 2:}
Add and subtract $r_{\theta_0}(\vect x,\vect Z_i)$ inside the posterior integral:
\begin{align}
\label{eq:B-split}
  (B)
  &=
  \int
    \Biggl\{
      \frac{1}{n}\sum_{i=1}^n
      \Bigl(
        r_{\theta_Y}(\vect x,\vect Z_i)-r_{\theta_0}(\vect x,\vect Z_i)
      \Bigr)
    \Biggr\}\,\widetilde\Pi_n(d\theta)
  \nonumber\\
  &\quad+\;
  \Biggl\{
    \frac{1}{n}\sum_{i=1}^n r_{\theta_0}(\vect x,\vect Z_i)
    -
    \E_{\theta_0}\bigl[r_{\theta_0}(\vect x,\vect Z_i)\bigr]
  \Biggr\}.
\end{align}
The second bracket in \eqref{eq:B-split} converges to $0$ in
$L^1(P_{\theta_0})$, and hence in probability, by the oracle empirical-average
condition in Theorem~\ref{thm:bayes-mu-consistency}. For the first bracket, use
the triangle inequality and integrate with respect to
$\widetilde\Pi_n(d\theta)$:
\begin{align}
  |(B)|
  &\le
  \int
    \frac{1}{n}\sum_{i=1}^n
    \bigl|r_{\theta_Y}(\vect x,\vect Z_i)-r_{\theta_0}(\vect x,\vect Z_i)\bigr|\,
  \widetilde\Pi_n(d\theta)
  \;+\;
  o_{P_{\theta_0}}(1).
\end{align}
By \eqref{eq:bayes-m-unif} in Assumption~\ref{ass:bayes-contraction-theta0},
the posterior integral converges to zero in $L^1(P_{\theta_0})$, hence also in probability.
Therefore $(B)=o_{P_{\theta_0}}(1)$.

Combining Steps 1-2 in \eqref{eq:bayes-decomp} yields
\begin{align}
  \mu_n^{\mathrm{pp}}(\vect x)
  \xrightarrow{P_{\theta_0}}
  \mu(\vect x)
  \qquad\text{as } n,m\to\infty \text{ and } \varepsilon_{n,m}\to0,
\end{align}
which proves the asserted consistency.

For the bias bound, take expectations in \eqref{eq:bayes-decomp} and use
the triangle inequality together with \eqref{eq:A-bound}:
\begin{align}
  \bigl|
    \E_{\theta_0}[\mu_n^{\mathrm{pp}}(\vect x)]
    - \mu(\vect x)
  \bigr|
  \le
  \E_{\theta_0}|(A)|
  +
  \E_{\theta_0}|(B)|.
\end{align}
By \eqref{eq:A-bound}, $\E_{\theta_0}|(A)|\le L(\vect x)\varepsilon_{n,m}+o(1)$.
Moreover, Assumption~\ref{ass:bayes-envelope} and the calibration condition imply
$\E_{\theta_0}|(B)|=o(1)$: the posterior-calibration component is $o(1)$ in
$L^1(P_{\theta_0})$, and the empirical-average component has expectation
$o(1)$ by the oracle empirical-average condition.
Therefore,
\begin{align}
  \bigl|
    \E_{\theta_0}[\mu_n^{\mathrm{pp}}(\vect x)]
    - \mu(\vect x)
  \bigr|
  \le
  L(\vect x)\,\varepsilon_{n,m} + o(1),
\end{align}
as claimed.
\end{proof}

Under the assumptions of Theorem~\ref{thm:bayes-mu-consistency}, posterior-predictive average treatment effects satisfy
\begin{align}
  \Delta_n^{\mathrm{pp}}(\vect x,\vect x')
  :=
  \mu_n^{\mathrm{pp}}(\vect x')-\mu_n^{\mathrm{pp}}(\vect x)
  \xrightarrow{p}
  \mu(\vect x')-\mu(\vect x)
  =
  \Delta(\vect x,\vect x'),
\end{align}
and the bias bound follows by triangle inequality with Lipschitz constants
$L(\vect x)$ and $L(\vect x')$.

\section{Additional synthetic-grid diagnostics}
\label{sec:appendix-synthetic-grid}

This appendix records supporting diagnostics for the synthetic grid in
Section~\ref{sec:experiments}. Figure~\ref{fig:appendix-synthetic-grid-diagnostics}
shows the full set of grid sweeps. The signal ratio and outcome validity panels are the main paper diagnostics; the localisation and nuisance-rank panels are used as sensitivity checks. They show that performance depends on the geometry of the assignment problem, not only on a one-dimensional signal-to-noise
summary.

\begin{figure}[htbp]
    \centering
    \includegraphics[width=\linewidth]{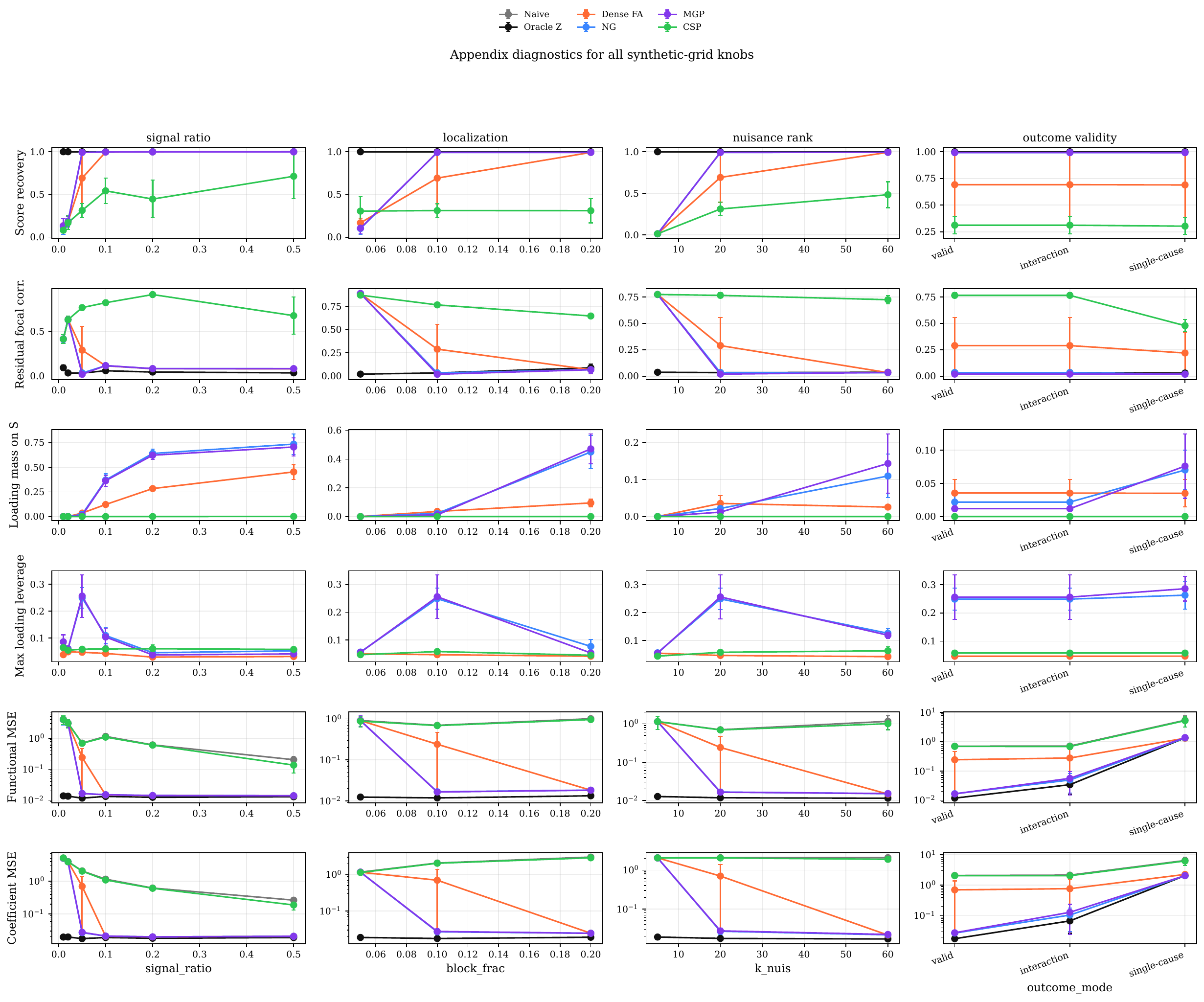}
    \caption{Full synthetic-grid diagnostics. Rows report score recovery,
    residual focal dependence, loading mass on the confounder block, maximum
    loading leverage, functional MSE, and coefficient MSE. Columns vary the
    nominal confounder signal, block localisation, nuisance rank, and outcome
    validity.}
    \label{fig:appendix-synthetic-grid-diagnostics}
\end{figure}

Figures~\ref{fig:appendix-scenario-a-dgp} and
\ref{fig:appendix-scenario-b-dgp} show the two original fixed calibrations used
to motivate the grid. They are retained here to make the assignment geometry
visible. In both designs, the focal causes lie in a small confounder block
\(S\), while high-variance nuisance structure lies mostly in \(S^c\). Scenario~B
reduces the confounder strength and block size relative to Scenario~A.

\begin{figure}[htbp]
    \centering
    \includegraphics[width=0.95\linewidth]{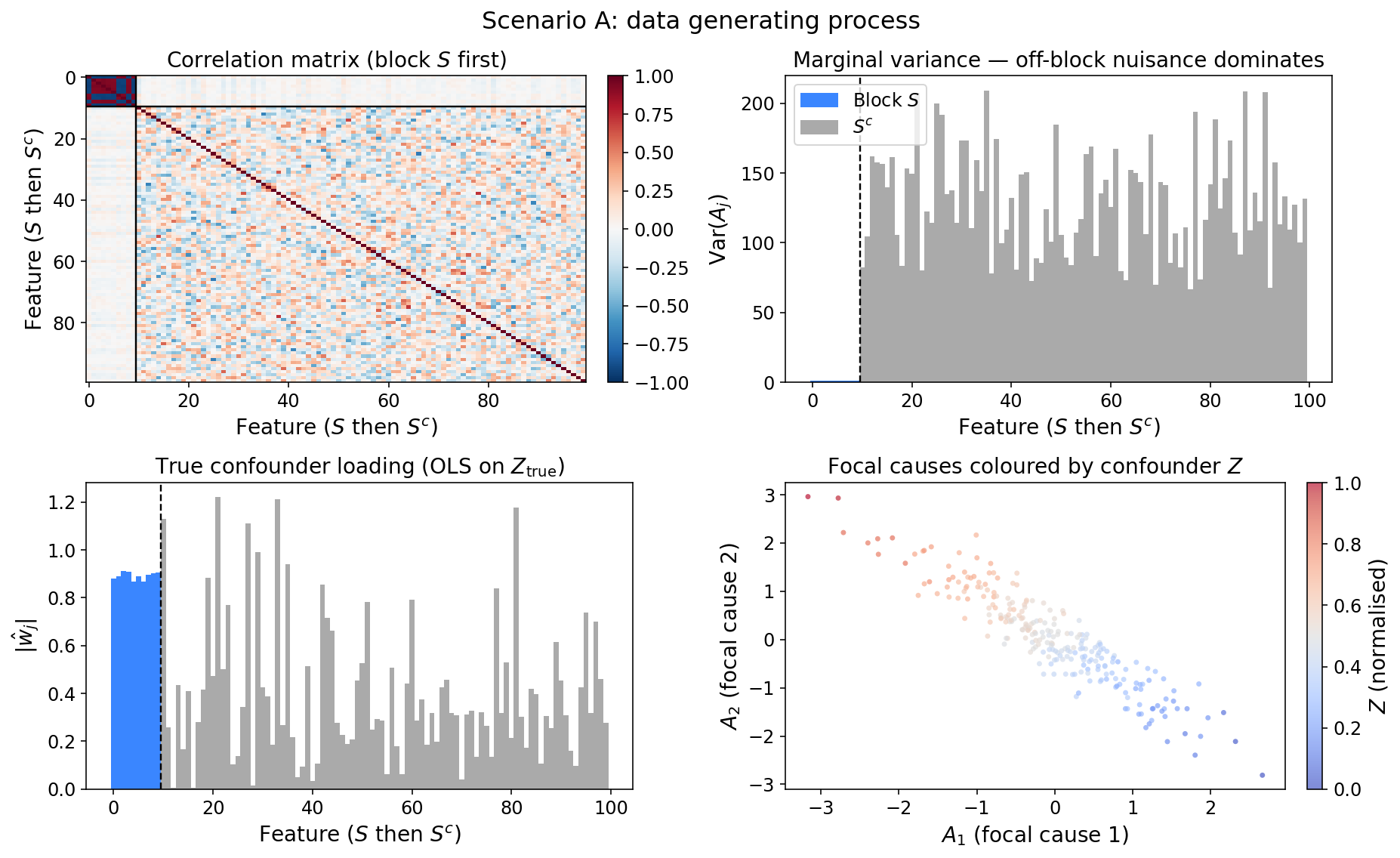}
    \caption{Original Scenario A geometry. The confounding signal is localised
    on a small block containing the focal causes, while the largest marginal
    variances arise in off-block nuisance directions.}
    \label{fig:appendix-scenario-a-dgp}
\end{figure}

\begin{figure}[htbp]
    \centering
    \includegraphics[width=0.95\linewidth]{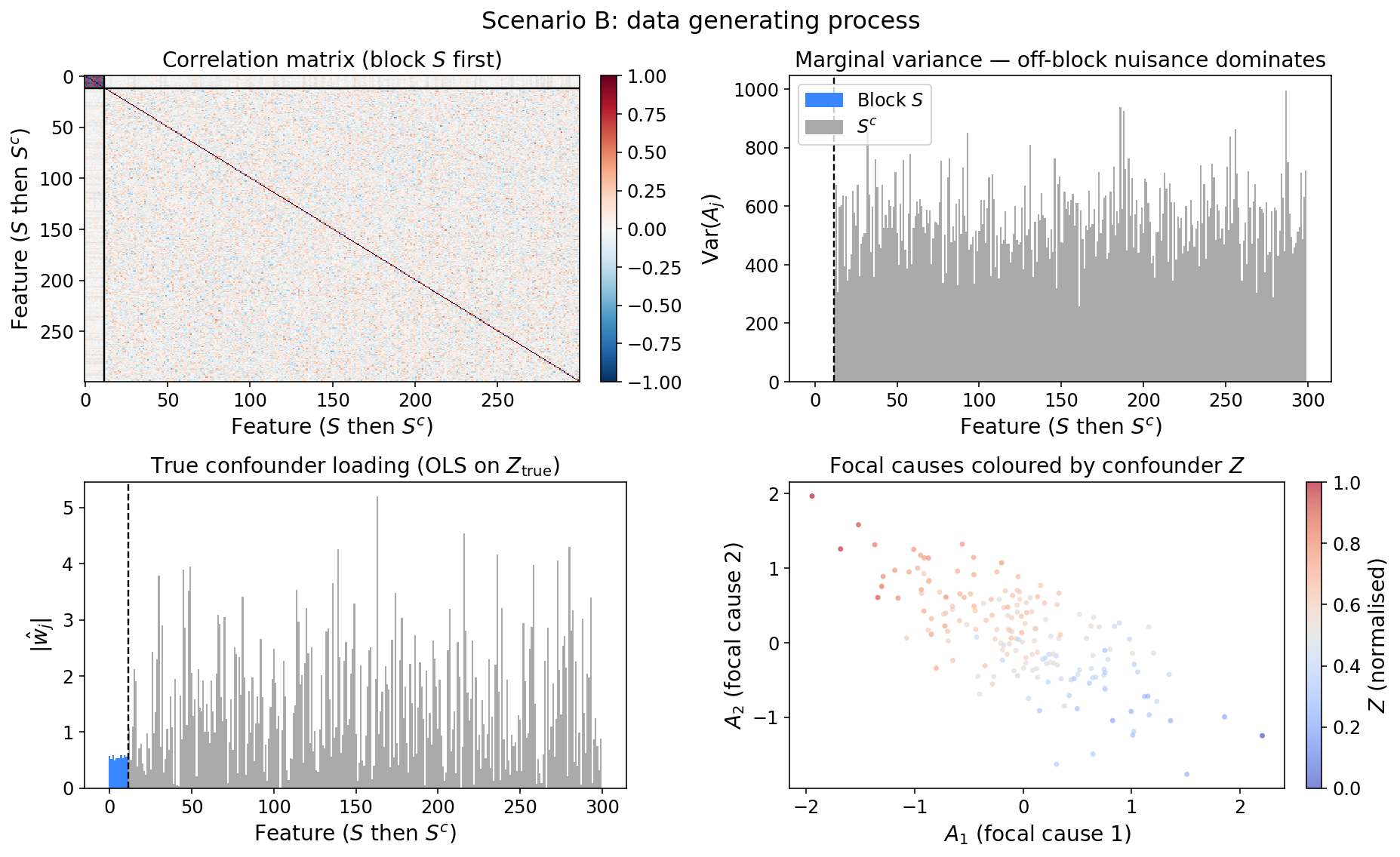}
    \caption{Original Scenario B geometry. The confounder block is smaller and
    weaker than in Scenario A, and the off-block nuisance structure is richer.}
    \label{fig:appendix-scenario-b-dgp}
\end{figure}

\section{Additional ADNI diagnostics}
\label{sec:appendix-adni-diagnostics}

This appendix records two supporting checks for the ADNI example in
Section~\ref{subsec:adni-example}. First, Table~\ref{tab:appendix-adni-apoe}
reports the APOE4/education stress test. The analysis uses an enriched CN/AD
subset in which AD cases have at least one APOE4 allele and lower education,
whereas CN controls have no APOE4 allele and higher education. Sex and
intracranial volume are included as common outcome covariates. The factor
models are fit after withholding APOE4 and education, using standardized age,
log amyloid-\(\beta\), log tau, log \(p\)-tau, and FDG as inputs. This design
asks whether the remaining baseline markers can reproduce the direct
APOE4/education adjustment.

\begin{table}[htbp]
\centering
\small
\setlength{\tabcolsep}{7pt}
\caption{APOE4/education stress test in ADNI (\(n=146\)). The benchmark adjusts
directly for APOE4 and education in addition to sex and intracranial volume.
Substitute confounder rows withhold APOE4 and education from the factor model.}
\label{tab:appendix-adni-apoe}
\begin{tabular}{lccc}
\toprule
Method & Posterior mean & 95\% CI & Collapsed factors \\
\midrule
APOE4/education benchmark & $-2356$ & $(-3012,\ -1702)$ & -- \\
Sex/ICV only & $-1670$ & $(-1922,\ -1431)$ & -- \\
Dense FA & $-1705$ & $(-2154,\ -1240)$ & p-tau \\
Sparse FA (NG) & $-1724$ & $(-2214,\ -1201)$ & tau \\
Sparse FA (MGP) & $-1716$ & $(-2219,\ -1234)$ & p-tau \\
Sparse FA (CSP) & $-1722$ & $(-2201,\ -1247)$ & p-tau \\
\bottomrule
\end{tabular}
\end{table}

The substitute adjustments in Table~\ref{tab:appendix-adni-apoe} move only
modestly away from the sex/ICV-only contrast toward the direct
APOE4/education benchmark. The sparse FA posterior means are slightly closer to the benchmark than dense FA, but all methods leave most of the APOE4/education shift unrecovered, and each factor model has a factor dominated by tau or phosphorylated tau. This stress test complements the CSF-withheld
result in the main text: the available proxies carry some disease-related structure, but not enough shared information to reproduce every clinically relevant adjustment.

Second, we considered a sensitivity analysis in which the CSF-withheld factor model also included non-hippocampal MRI volumes, such as whole-brain, entorhinal, fusiform, and middle-temporal volumes. This proxy set is more outcome proximal. It over-adjusted the diagnostic contrasts relative to the direct CSF benchmark: the mean absolute gaps from the CSF benchmark were approximately \(267\)-\(296\) \(\mathrm{mm}^3\), and mean gap reductions were negative for all factor methods. We therefore treat the demographic/genetic/FDG proxy
set in Table~\ref{tab:adni-csf-recovery} as the primary CSF-withheld benchmark and the MRI-heavy specification as a cautionary sensitivity check.

\bibliographystyle{plainnat}
\bibliography{references}
 
\end{document}